\documentclass[runningheads,orivec]{llncs}

\usepackage[T1]{fontenc}
\usepackage[colorlinks=true]{hyperref}
\usepackage[svgnames,x11names,table]{xcolor}
\usepackage{tabularray}

\usepackage{float}
\usepackage{graphicx}
\usepackage{subcaption}
\usepackage{tikz}
\usepackage{tikz-cd}
\usepackage{circuitikz}
\usetikzlibrary{petri,positioning, calc,decorations.pathmorphing,shapes,cd,arrows}

\usepackage{caption}
\usepackage{multirow}
\usepackage{enumitem}

\usepackage{url}
\usepackage{siunitx}
\usepackage{amssymb,mathrsfs,nccmath,bbm}
\usepackage{stmaryrd}
\usepackage{centernot}
\newtheorem{notation}{Notation}

\hypersetup{
  urlcolor=Blue4,
  citecolor=Blue4,
  linkcolor=Blue4
}

\newcommand{\arrowflow}{
\tikz \draw[-{to}] (-1pt,0) -- (1pt,0);
}

\newcommand{\flow}[4]{ 
	\draw[#3] (#1) to [#4] node [pos=0.5] {\arrowflow} (#2);
}
\newcommand{\flowdiag}[5]{ 
	\draw[#3] (#1) to [#4] node [#5] {\arrowflow} (#2);
}

\tikzset{fill fraction/.style n args={1}{path picture={
 \fill[black] (path picture bounding box.south west) rectangle
 ($(path picture bounding box.north west)!0.5!(path picture bounding box.north
 east)$);}}
}
\newcommand{\qmatch}[4]{ 
	\node[draw,fill=white,fill fraction={black}{0.5},inner sep=1.2pt,minimum size=3pt,label=above:{\scriptsize #4}] (#1) at (#2,#3) {};
}
\newcommand{\dmatch}[5]{ 
	\node[circle,draw,fill=white,fill fraction={black}{0.5},inner sep=1.2pt,minimum size=3pt,label=#5:{\scriptsize #4}] (#1) at (#2,#3) {};
}
\newcommand{\qinplain}[4]{ 
\node[draw=black,fill=white,inner sep=0pt,minimum size=3pt,label=left:\scriptsize #4] (#1) at (#2,#3) {};
}
\newcommand{\qin}[4]{ 
\flow{{#2,#3}}{$(#2,#3)+(0.8,0)$}{dashed}{};
\qinplain{#1}{#2}{#3}{#4};
}
\newcommand{\dinplain}[4]{ 
\node[circle,draw=black,fill=white,inner sep=0pt,minimum size=3pt,label=left:\scriptsize #4] (#1) at (#2,#3) {};
}
\newcommand{\din}[4]{ 
\flow{{#2,#3}}{$(#2,#3)+(0.8,0)$}{}{};
\dinplain{#1}{#2}{#3}{#4};
}
\newcommand{\qoutplain}[4]{ 
\node[fill=black,inner sep=0pt,minimum size=3pt,label=right:\scriptsize #4] (#1) at (#2+0.8,#3) {};
}
\newcommand{\qout}[4]{ 
\flow{{#2,#3}}{$(#2,#3)+(0.8,0)$}{dashed}{};
\qoutplain{#1}{#2}{#3}{#4};
}
\newcommand{\doutplain}[4]{ 
\node[circle,fill=black,inner sep=0pt,minimum size=3pt,label=right:\scriptsize #4] (#1) at (#2+0.8,#3) {};
}
\newcommand{\dout}[4]{ 
\flow{{#2,#3}}{$(#2,#3)+(0.8,0)$}{}{};
\doutplain{#1}{#2}{#3}{#4};
}

\newcommand{\computonPrimitive}[5]{ 
\draw[draw=black,fill=white] (#1,#2) rectangle ++(#3,#4) node[pos=0.5]{#5};
}
\newcommand{\computonComposite}[4]{ 
\draw[draw=black,fill=black!2] (#1,#2) rectangle ++(#3,#4);
}

\begin{document}

\title{Compositional Control-Driven Boolean Circuits}
\author{Damian Arellanes\inst{1}}
\authorrunning{D. Arellanes}
\institute{Lancaster University, Lancaster LA1 4WA, United Kingdom
\email{damian.arellanes@lancaster.ac.uk}}

\maketitle              

\begin{abstract}
Boolean circuits abstract away from physical details to focus on the logical structure and computational behaviour of digital components. Although such circuits have been studied for many decades, compositionality has been widely ignored or examined in an informal manner, which is a property for combining circuits without delving into their internal structure, while supporting modularity and formal reasoning. In this paper, we address this longstanding theoretical gap by proposing colimit-based operators for compositional circuit construction. We define separate operators for forming sequential, parallel, branching and iterative circuits. As composites encapsulate explicit control flow, a new model of computation emerges which we refer to as (families of) \emph{control-driven Boolean circuits}. We show how this model is at least as powerful as its classical counterpart. In other words, it is able to non-uniformly compute any Boolean function on inputs of arbitrary length. 

\keywords{Boolean circuits \and algebraic composition \and explicit control flow \and non-uniform computation \and models of high-level computation.}
\end{abstract}

\section{Introduction}
\label{sec:Introduction}

Families of Boolean circuits define a non-uniform model of computation in the sense that a collection of circuits can be used to solve a particular problem (e.g., deciding a formal language) \cite{clote_boolean_2002}. To realise any Boolean function, a functionally complete basis of logical connectives is required for circuit construction such as $\{\land,\lor,\lnot\}$, $\{\implies,\lnot\}$ or $\{\uparrow\}$ \cite{wegener_complexity_1987}. 

Although this model of computation can decide arbitrary languages---including undecidable ones---by hard-coding solutions for each input length, it does not provide theoretical underpinnings to enable the compositional construction of Boolean circuits via well-defined algebraic operators that satisfy closure. By composition, we mean forming large complex circuits from simpler ones, without delving into the internals of the composed entities \cite{ghica_complete_2025}. Endowing circuits with compositional capabilities can enhance formal analysis and can allow treating circuits as modular, self-similar black-boxes that can be reused across a wide range of application domains \cite{giannakopoulou_compositional_2018}.

For formal analysis, compositional verification is perhaps the most prominent consequent technique, which allows proving that a property $\phi$ holds for a composite $C$ whenever $\phi$ holds for all the constituents of $C$ and $\phi$ is preserved across all the composition levels underlying $C$ \cite{dardik_recomposition_2024}. 

Although composition can be realised in many different ways, in this paper we follow up on our idea that every (high- or low-level) computing device exhibits either implicit or explicit control flow, no matter how data is exchanged \cite{arellanes_compositional_2026,arellanes_models_2025}. Accordingly, we propose a theory for compositionally constructing Boolean circuits through colimit-based composition operators that define explicit control flow for the computation of Boolean functions in some specific order. We formalise distinct operators for forming sequential, parallel, branching and iterative circuits. Sequencing and parallelising are obvious in terms of behaviour. While branching captures non-deterministic choice, iteration allows the repeated execution of some circuit. For sequencing, parallelising and branching we show relevant algebraic properties in terms of identity, commutativity and associativity. We also show how every Boolean function can be computed by the new model we introduce, which we refer to as (families of) \emph{control-driven Boolean circuits}. 

The rest of the paper is organised as follows. Section \ref{sec:model} introduces the new class of circuits we propose, describes colimit-based operations for their compositional construction and presents operational semantics for circuit execution. Section \ref{sec:equivalence} discusses the relationship between control-driven circuits and their classical counterpart in terms of what can be computed. Section \ref{sec:related-work} analyses related work. Section \ref{sec:conclusions} outlines the final remarks and discusses future directions. 

\section{Control-Driven Boolean Circuits}
\label{sec:model}

The computon model \cite{arellanes_compositional_2026} provides a general framework to specify and compose high-level, control-driven computations which dictate the explicit order in which (low- or high-level) computations occur. Intuitively, it allows forming bipartite-graph-like structures (called computons) where nodes are computation units or typed ports, while edges are information flows from units to ports or vice versa, never between elements of the same kind. A unit encapsulates a low-level computation, whereas a port is a buffer for control signals or other typed data. 

In this paper, we move from the abstract to the concrete by studying a particular class of computons in which computation units are functions that establish an indexed family of NAND operators, while ports are variables holding control or Boolean values only (see Definition \ref{def:cunit}). Accordingly, the set of types on which this class operates consists of the unit type $\mathbbm{1}$ and the Boolean type $\mathbb{B}$. As high-level computations define control flow for the evaluation of NAND operators in some order, we refer to this class of computons as control-driven Boolean circuits (or circuits for short). Definition \ref{def:computon} captures this notion, which is a particular variation of the one given in \cite{arellanes_compositional_2026}.\footnote{Here, $\twoheadrightarrow$ denotes surjection. The unit type $\mathbbm{1}$ has only one inhabitant that represents a control signal. So, variables mapped to $\mathbbm{1}$ are called control variables. Intuitively, a circuit with control variables only at the interface, receives and produces control signals only, which is useful for joining/forking control from/to multiple circuits.}

\begin{definition}[Computation Unit]\label{def:cunit}
A computation unit is the indexed family ${(\uparrow_k)_{k\in \mathbb{N}}}$ where ${\uparrow_0=1}$ is a constant function and, for ${k>0}$, ${\uparrow_k}$ is the $k$-ary NAND operator ${\mathbb{B}^k\rightarrow\mathbb{B}}$ given by ${\uparrow_k(b_1,\ldots,b_k)=\lnot(b_1\land\cdots\land b_k)}$.
\end{definition}

\begin{definition}[Control-Driven Boolean Circuit] \label{def:computon}
A control-driven Boolean circuit $\lambda$ is a 10-tuple ${(V,U,I,O,\Sigma,s,t,\sigma,\tau,c)}$ where:
\begin{itemize}[noitemsep, topsep=1pt]
\item $V$ is a non-empty finite set of variables,
\item $U$ is a finite set of computation units,
\item $I$ is a finite set of input flows linking variables to computation units,
\item $O$ is a finite set of output flows linking computation units to variables,
\item ${\Sigma}\subseteq\{\mathbbm{1},\mathbb{B}\}$ is a non-empty finite set of types,
\item ${s\colon I\rightarrow V}$ is a function that specifies the source variable of each input flow,
\item ${\tau\colon I\twoheadrightarrow U}$ is a function that defines the target unit of each input flow,
\item ${\sigma\colon O\twoheadrightarrow U}$ is a function that specifies the source unit of each output flow,
\item ${t\colon O\rightarrow V}$ is a function defining the target variable of each output flow and
\item ${c\colon V\twoheadrightarrow\Sigma}$ is a function that assigns a type to each variable,
\end{itemize}
such that ${\sigma\restriction_{(c\circ t)^{-1}(\mathbbm{1})}}$ and ${\tau\restriction_{(c\circ s)^{-1}(\mathbbm{1})}}$ are surjective and there is at least one variable ${v \in V\setminus s(I)}$ and at least one variable ${w \in V\setminus t(O)}$ with ${c(v)=\mathbbm{1}=c(w)}$.
\end{definition}

\begin{notation}\label{not:port-types}
Given a circuit $\lambda$, we call ${v\in V}$ a control variable if ${c(v)=\mathbbm{1}}$; otherwise, we call it a Boolean variable. We write ${v\xrightarrow{i}u}$ to express ${s(i)=v}$ and ${\tau(i)=u}$ for some ${i\in I}$, some ${v\in V}$ and some ${u\in U}$. Likewise, we use ${u\xrightarrow{o}w}$ to denote ${\sigma(o)=u}$ and ${t(o)=w}$ for some ${o\in O}$ and some ${w\in V}$. With this, we write ${v_1\xrightarrow{\exists}v_n}$ to express that ${v_1\xrightarrow{i_1}u_1\xrightarrow{o_1}v_2\xrightarrow{i_2}\cdots\xrightarrow{i_{n-1}}u_{n-1}\xrightarrow{o_{n-1}}v_n}$ holds for ${v_1,v_2,\ldots,v_n\in V}$, ${u_1,\ldots,u_{n-1}\in U}$, ${i_1,i_2,\ldots,i_{n-1}\in I}$ and ${o_1,\ldots,o_{n-1}\in O}$.
\end{notation}

\begin{notation}\label{not:pre-post-sets}
If ${u \in U}$ is a unit of a circuit $\lambda$, ${\bullet u:=s(\tau^{-1}(u))}$ and ${u\bullet:=t(\sigma^{-1}(u))}$ denote the sets of variables connected to and from $u$, respectively. Similarly, for ${v\in V}$, ${v\bullet:=\tau(s^{-1}(v))}$ and ${\bullet v:=\sigma(t^{-1}(v))}$. From now on, we use natural numbers as subindices to distinguish between different circuits and their components. If a symbol for a circuit has no subindex, its components have no subindex either.
\end{notation}

A glance at Definition \ref{def:computon} reveals that a circuit $\lambda$ has neither dangling computation units nor dangling flows, as a result of imposing surjectivity on $\sigma$ and $\tau$ as well as totality on $s$ and $t$. Although $c$ can be non-surjective, we decided to treat it otherwise to express that a circuit operates on all its types. Notice in Definition \ref{def:computon} that the surjectivity conditions on ${\sigma\restriction_{(c\circ t)^{-1}(\mathbbm{1})}}$ and ${\tau\restriction_{(c\circ s)^{-1}(\mathbbm{1})}}$ enforce units to always have input and output flows attached to control variables. Although it is possible for $\lambda$ to have no computation units and no flows at all, the presence of typed variables is mandatory. In fact, the last two restrictions from Definition \ref{def:computon} explicitly state that there must be at least one variable with no incoming flows and at least one variable with no outgoing flows mapped to the unit type $\mathbbm{1}$. Implicitly, this means $\Sigma$ is ${\{\mathbbm{1}\}}$ or ${\{\mathbbm{1},\mathbb{B}\}}$, i.e., the unit type is always in $\Sigma$.

Variables with no incoming flows are called invars whereas variables with no outgoing flows are called outvars. Together, they form the interface of a circuit towards the external world, as captured by Definition \ref{def:interface}. As $\mathbbm{1}$ is the type that represents control signals, circuits always have control invars and control outvars. $\mathbb{B}$-typings are reserved for Boolean values ($0$ or $1$ for convenience). 

\begin{definition}[Circuit Interface]\label{def:interface}
The interface of a circuit $\lambda$ is a tuple ${(V^+,V^-)}$ where ${V^+}$ and ${V^-}$ are the sets ${V\setminus t(O)}$ and ${V\setminus s(I)}$, respectively.
\end{definition}

When a circuit exhibits a sequence of flows that take each invar to some outvar, we say soundness is satisfied. This property is formalised in Definition \ref{def:sound}.

\begin{definition}[Sound Circuit]\label{def:sound}
A circuit is sound if every ${v\in s(I)\cup V^+}$ satisfies ${v\xrightarrow{\exists}w}$ for some ${w\in V^-}$.
\end{definition}

More technically, Definition \ref{def:sound} applies not only to invars but to any variable with outgoing flows. We consider these two cases because ${V^+\subseteq s(I)}$ does not hold in general given that $V^+$ and $V^-$ are not necessarily disjoint. A variable ${v\in V^+\cap V^-}$ is called inoutvar and, according to Definition \ref{def:interface}, it has neither incoming nor outgoing flows. Thus, a circuit with at least one inoutvar does not satisfy soundness.

\emph{Trivial Circuits} are a class with all variables in ${V^+\cap V^-}$. Although they are not sound, they are needed for the coherence of our theory. Figure \ref{fig:computon-trivial-primitive}(a) shows that, apart from having all inoutvars, they have no units and no flows at all. An example of a trivial circuit with $2$ control inoutvars and $1$ Boolean inoutvar is shown in Figure \ref{fig:example-basic}(a)

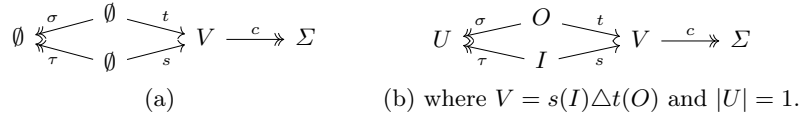
\begin{figure}[H]
\centering
\subcaptionbox{}
{
\begin{tikzcd}[ampersand replacement=\&, row sep=-0.5em]
 \& \emptyset \arrow[dl,twoheadrightarrow,"\sigma"']\arrow[dr, "t"] \& \&  \\
\emptyset \& \& V \arrow[r,twoheadrightarrow,"c"] \& \Sigma \\
 \& \emptyset \arrow[ul,twoheadrightarrow,"\tau"]\arrow[ur, "s"'] \& \& \\
\end{tikzcd}
}
\subcaptionbox{where ${V = s(I) \triangle t(O)}$ and ${|U|=1}$.}[0.55\textwidth]
{
\begin{tikzcd}[ampersand replacement=\&, row sep=-0.5em]
 \& O \arrow[dl,twoheadrightarrow,"\sigma"']\arrow[dr, "t"] \& \&  \\
U \& \& V \arrow[r,twoheadrightarrow,"c"] \& \Sigma \\
 \& I \arrow[ul,twoheadrightarrow,"\tau"]\arrow[ur, "s"'] \& \& \\
\end{tikzcd}
}
\caption{The set-valued diagrams (a) and (b) capture the structure of trivial and primitive circuits, respectively. Here, $\triangle$ denotes symmetric set difference.}
\label{fig:computon-trivial-primitive}
\end{figure}

\begin{remark}\label{rem:unit}
By Definition \ref{def:computon}, we know that the set of variables of any circuit cannot be empty and that there must be at least one control variable. Accordingly, the simplest circuit is a trivial one with exactly one control inoutvar. Up to isomorphism, this circuit is referred to as the \emph{unit circuit}, denoted $\Lambda$.
\end{remark}

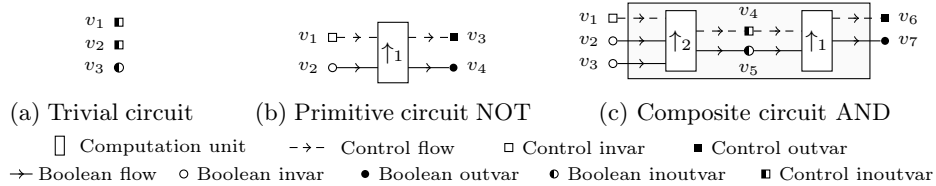
\begin{figure}[H] 
 \subcaptionbox{Trivial circuit}[0.2\textwidth]
 {
\begin{tikzpicture}
\node[draw,fill=white,fill fraction={black}{0.5},inner sep=1.2pt,minimum size=3pt,label=left:{\scriptsize $v_1$}] (q1) at (0,0.7) {};
\node[draw,fill=white,fill fraction={black}{0.5},inner sep=1.2pt,minimum size=3pt,label=left:{\scriptsize $v_2$}] (qi) at (0,0.4) {};
\node[circle,draw,fill=white,fill fraction={black}{0.5},inner sep=1.2pt,minimum size=3pt,label=left:{\scriptsize $v_3$}] (dj) at (0,0.1) {};
\node(del2) at (0,0){};
\end{tikzpicture}
}\hspace{0.45cm}
 \subcaptionbox{Primitive circuit NOT}[0.3\textwidth]
 {
 \begin{tikzpicture}
 \computonPrimitive{0.6}{0}{0.4}{0.8}{$\uparrow_1$}

 \node[draw=black,fill=white,inner sep=0pt,minimum size=3pt,label=left:{\scriptsize $v_1$}] (1q0) at (0,0.6) {};\flow{1q0}{$(1q0)+(0.6,0)$}{dashed}{};
 \dinplain{1i1}{0}{0.2}{$v_2$};\flow{1i1}{$(0,0.2)+(0.6,0)$}{}{};
 
 \node[fill=black,inner sep=0pt,minimum size=3pt,label=right:{\scriptsize $v_3$}] (2q1) at (1.6,0.6) {};\flow{{1,0.6}}{$(1,0.6)+(0.6,0)$}{dashed}{};
 \doutplain{1o1}{0.8}{0.2}{$v_4$};\flow{{1,0.2}}{1o1}{}{};
 \end{tikzpicture}
 }
 \subcaptionbox{Composite circuit AND}
 {
 \begin{tikzpicture}
\computonComposite{0.4}{-0.1}{3.2}{1};
\qmatch{3q}{2}{0.53}{$v_4$};\flow{{1.3,0.53}}{3q}{dashed}{};\flow{3q}{{2.7,0.53}}{dashed}{};
\dmatch{3d}{2}{0.27}{$v_5$}{below};\flow{{1.3,0.27}}{3d}{}{};\flow{3d}{{2.7,0.27}}{}{};

\computonPrimitive{0.9}{0}{0.4}{0.8}{$\uparrow_2$}

\qinplain{1q1}{0.2}{0.7}{$v_1$};\flow{1q1}{{0.9,0.7}}{dashed}{};
\dinplain{1i1}{0.2}{0.4}{$v_2$};\flow{1i1}{{0.9,0.4}}{}{};
\dinplain{1i2}{0.2}{0.1}{$v_3$};\flow{1i2}{{0.9,0.1}}{}{};

\computonPrimitive{2.7}{0}{0.4}{0.8}{$\uparrow_1$}
\qoutplain{2q1}{3}{0.7}{$v_6$};\flow{{3.1,0.7}}{2q1}{dashed}{};
\doutplain{2o1}{3}{0.4}{$v_7$};\flow{{3.1,0.4}}{2o1}{}{};
 \end{tikzpicture}
 }
 {
\begin{tikzpicture}
\begin{scope}[xshift=0.6cm,yshift=-0.15cm]
\draw[draw=black,fill=white] (0,0.4) rectangle ++(0.12,0.25);
\node[inner sep=0pt,minimum size=0pt,label=right:{\scriptsize Computation unit}] at (0.2,0.5) {};
\draw[densely dashed] (3.1,0.5) to node[pos=0.5]{\arrowflow} (3.6,0.5);
\node[inner sep=0pt,minimum size=3pt,label=right:{\scriptsize Control flow}] at (3.6,0.5) {};
\node[draw=black,fill=white,inner sep=0pt,minimum size=3pt,label=right:{\scriptsize Control invar}] at (6,0.5) {};
\node[fill=black,inner sep=0pt,minimum size=3pt,label={right:{\scriptsize Control outvar}}] at (8.5,0.5) {};
\end{scope}

\begin{scope}
\draw (0,0) to node[pos=0.5]{\arrowflow} (0.3,0);
\node[inner sep=0pt,minimum size=3pt,label=right:{\scriptsize Boolean flow}] at (0.2,0){};
\node[circle,draw=black,fill=white,inner sep=0pt,minimum size=3pt,label={right:{\scriptsize Boolean invar}}] at (2.3,0) {};
\node[circle,fill=black,inner sep=0pt,minimum size=3pt,label={right:{\scriptsize Boolean outvar}}] at (4.7,0) {};
\node[circle,draw,fill=white,fill fraction={black}{0.5},inner sep=0pt,minimum size=3pt,label=right:{\scriptsize Boolean inoutvar}] at (7.2,0) {};
\node[draw,fill=white,fill fraction={black}{0.5},inner sep=0pt,minimum size=3pt,label=right:{\scriptsize Control inoutvar}] at (10,0) {};
\end{scope}
\end{tikzpicture}
 }%
 \caption{A trivial circuit and two control-driven logical connectives. Intuitively, (a) has no behaviour because it has no units. Although variables are labelled with some type through a function $c$, we label them with identifiers for convenience.} 
 \label{fig:example-basic}
\end{figure}

Besides trivial objects, we also have \emph{primitive circuits} which are sound entities with only invars and outvars, attached to exactly one unit (see Proposition~\ref{prop:functional-sound}). The structure of a primitive circuit is given by the set-valued diagram from Figure \ref{fig:computon-trivial-primitive}(b), which has constraints on the sets of units and variables. An example with $1$ control invar, $1$ Boolean invar, $1$ control outvar and $1$ Boolean outvar is displayed in Figure \ref{fig:example-basic}(b).

\begin{proposition}\label{prop:functional-sound}
Every primitive circuit is sound.
\end{proposition}
\begin{proof}
Assuming that $\lambda$ is a primitive circuit with ${v\in s(I)\cup V^+}$, we first show ${s(I)=V^+}$ as follows: ${v \in s(I) \iff v\notin t(O)}$ since ${V=s(I)\triangle t(O) \iff v \in V^+}$ by Definition \ref{def:interface}. So, it suffices to show only for ${v\in s(I)}$ which entails there exists some ${i\in I}$ where ${s(i)=v}$. As $\tau$ is total surjective and ${|U|=1}$, ${\tau(i)=u}$ for the unique ${u\in U}$. Similarly, as $\sigma$ is total surjective, there must be some ${o\in O}$ such that ${\sigma(o)=u}$ and, therefore, ${t(o)=w}$ must hold for some ${w\in V}$ by the totality of $t$. Using the fact ${V=s(I)\triangle t(O)}$ and Definition \ref{def:interface}, we obtain ${w\in t(O)\iff w\notin s(I)\iff w\in V^-}$. As we have just deduced the existence of ${v\xrightarrow{i}u\xrightarrow{o}w}$, using Notation \ref{not:port-types} we conclude ${v\xrightarrow{\exists}w}$ for ${v\in s(I)\cup V^+}$ and ${w\in V^-}$.\qed
\end{proof}

Primitive circuits are important for our theory since each of them encapsulates a family of NAND operators. Operationally, a primitive circuit does not compute the whole family, but it chooses one operator at a time depending on the number $k$ of Boolean invars. For example, in Figure \ref{fig:example-basic}(b), a unary NAND operator would be chosen from the family established by the unique computation unit. Section \ref{sec:dynamics} describes operational semantics in more detail.

Interestingly, when ${k=1}$, NAND behaves as NOT as per Definition \ref{def:cunit}. When ${k>1}$, it might seem that NAND could lead to inconsistencies, given that there is no explicit invar order. However, this is not entirely true since NAND is built upon AND, an operation that satisfies commutativity. Although NAND maps a number of inputs to a single output, primitive circuits are able to produce multiple outputs from multiple inputs. In Section \ref{sec:dynamics}, we will see that all Boolean outvars take the same value from the chosen operator. This feature can be particularly relevant to replicate values onto multiple circuits (for subsequent processing). 

Overall, trivial and primitive circuits are the basic building blocks to compositionally form more complex circuits which encapsulate an explicit control flow structure. Such complex entities are called composite circuits and can be sequential, parallel, branching or iterative. Before introducing them, we present the notion of a circuit morphism which, intuitively, allows embedding a circuit into another (see Definition \ref{def:morphism}). In our theory, there are distinguished morphisms called in- and out-adjoints, which identify all the invars and all the outvars of a circuit $\lambda$, respectively (see Definition \ref{def:adjoints}). As the domain of all the in-adjoints of $\lambda$ are isomorphic to each other, we say ``the'' in-adjoint of $\lambda$. The same holds for out-adjoints. 

\begin{definition}[Circuit Morphism]\label{def:morphism}
A circuit morphism ${\alpha\colon\lambda_1\rightarrow\lambda_2}$ is a tuple $(\alpha_V,\alpha_U,\alpha_I,\alpha_O,\alpha_\Sigma)$ of total functions ${\alpha_V\colon V_1\rightarrow V_2}$, ${\alpha_U\colon U_1\rightarrow U_2}$, ${\alpha_I\colon I_1\rightarrow I_2}$, ${\alpha_O\colon O_1\rightarrow O_2}$ and ${\alpha_\Sigma\colon\Sigma_1\hookrightarrow\Sigma_2}$, satisfying the condition ${\vec{i}(\alpha)\cup\vec{o}(\alpha)\subseteq V_1^+\cup V_1^-}$ with ${\vec{i}(\alpha):=\{v\in V_1\mid\bullet\alpha_V(v)\setminus\alpha_V(\bullet v)\neq\emptyset\}}$ and $\vec{o}(\alpha):=\{v\in V_1\mid\alpha_V(v)\bullet\setminus\alpha_V(v\bullet)\neq\emptyset\}$, and making the following diagram commute:
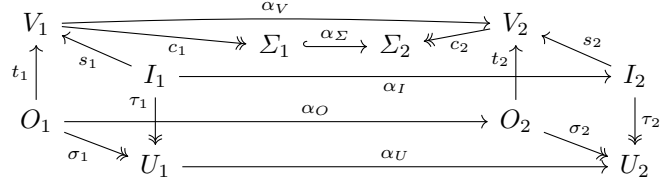
\[
\begin{tikzcd}
 V_1 \arrow[rrrr,bend left=2,"\alpha_V"]\arrow[drr,twoheadrightarrow,"c_1"'{xshift=15pt,yshift=-1pt}] & & & & V_2 \arrow[dl,twoheadrightarrow,"c_2"{xshift=-5pt}] & \\
 [-0.9cm]
 & & \Sigma_1 \arrow[r,hook,"\alpha_\Sigma"]  & \Sigma_2 & & \\
 [-0.8cm]
  & I_1 \arrow[rrrr,"\alpha_I"']\arrow[dd,twoheadrightarrow,"\tau_1"'{yshift=7pt}]\arrow[uul,"s_1"{xshift=3pt}] & & & & I_2 \arrow[dd,twoheadrightarrow,"\tau_2"]\arrow[uul,"s_2"'{yshift=-1pt}] \\
 [-0.6cm] 
 O_1 \arrow[rrrr,"\alpha_O"{xshift=15pt}]\arrow[uuu,"t_1"]\arrow[dr,twoheadrightarrow,"\sigma_1"'] & & & & O_2 \arrow[uuu,"t_2"{yshift=5pt}]\arrow[dr,twoheadrightarrow,"\sigma_2"{xshift=-5pt,yshift=1pt}] & \\
 [-0.6cm]
  & U_1 \arrow[rrrr,"\alpha_U"] & & & & U_2 \\
\end{tikzcd}
\]
\end{definition}

\begin{remark}
If $\alpha$ and $\beta$ are circuit morphisms, their composite ${\alpha\circ\beta}$ is given by:
\begin{small}
\[
(\alpha_V,\alpha_U,\alpha_I,\alpha_O,\alpha_\Sigma)\circ (\beta_V,\beta_U,\beta_I,\beta_O,\beta_\Sigma):=(\alpha_V\circ\beta_V,\alpha_U\circ\beta_U,\alpha_I\circ\beta_I,\alpha_O\circ\beta_O,\alpha_\Sigma\circ\beta_\Sigma)
\]
\end{small}
\end{remark}

\begin{definition}[Adjoint Morphism]\label{def:adjoints}
An adjoint ${\lambda^\square}$ of a circuit $\lambda$ is a circuit monomorphism ${\lambda_0\rightarrow\lambda}$ where $\lambda_0$ is a trivial circuit with ${\lambda^\square(V_0)=V^\square}$ and ${\square=\{+,-\}}$. If ${\square=+}$, it is called \emph{in-adjoint}; otherwise, it is called \emph{out-adjoint}.
\end{definition}

The conditions imposed on a circuit morphism ${\alpha\colon\lambda_1\rightarrow\lambda_2}$ enforce $\lambda_1$ to interact with $\lambda_2$ only at the boundaries. More precisely, ${\vec{i}(\alpha)\subseteq V_1^+\cup V_1^-}$ implies that, if ${v_1\in V_1}$ is mapped to a variable ${v_2\in V_2}$ connected from computation units not considered by ${\alpha_U}$, then $v_1$ must necessarily be either invar or outvar in $\lambda_1$. The condition ${\vec{o}(\alpha)\subseteq V_1^+\cup V_1^-}$ is similar, but it applies to $\lambda_2$-variables connected to computation units outside the image of $\alpha_U$. These two conditions ensure that invars (or outvars) of $\lambda_2$ in the image of $\alpha$ are always mapped from invars (or outvars) of $\lambda_1$ (see Proposition \ref{prop:interface-preservation}). Although it is possible to demote interface variables to inoutvars, the latter can never be promoted to invars or outvars in $\lambda_2$. Variables cannot change their type through $\alpha$ by the fact ${\alpha_\Sigma}$ is an inclusion.

\begin{proposition}\label{prop:interface-preservation}
For any circuit morphism ${\alpha\colon\lambda_1\rightarrow\lambda_2}$, we have ${\alpha^{-1}_V(V_2^+)\subseteq V_1^+}$ and ${\alpha^{-1}_V(V_2^-)\subseteq V_1^-}$.
\end{proposition}
\begin{proof}
We only show ${\alpha^{-1}_V(V_2^+)\subseteq V_1^+}$ by contrapositive since the proof of the other is analogous. For this, assume ${v_1\in V\setminus V_1^+}$ so there is some ${o_1\in O_1}$ with ${t_1(o_1)=v_1}$ (see Definition \ref{def:interface}). Using the commutativity property of $\alpha$, we have ${t_2(\alpha_O(o_1))=\alpha_V(t_1(o_1))=\alpha_V(v_1)}\implies{\alpha_V(v_1)\notin V_2^+}$ (according to Definition \ref{def:interface}) ${\implies v_1\notin\alpha_V^{-1}(V_2^+)}$. As ${v_1\notin V_1^+\implies v_1\notin\alpha_V^{-1}(V_2^+)}$ is logically equivalent to ${v_1\in\alpha_V^{-1}(V_2^+)\implies v_1\in V_1^+}$, we conclude ${\alpha^{-1}_V(V_2^+)\subseteq V_1^+}$, as required.\qed
\end{proof}

As all computation units establish the same family of NAND operators (see Definition \ref{def:cunit}), a circuit morphism always ensures behaviour preservation. In other words, no matter the mapping done by the $U$-component of a circuit morphism $\alpha\colon\lambda_1\rightarrow\lambda_2$, it is always true that $u_1\in U_1$ is essentially the same as $\alpha_U(u_1)\in U_2$.

Circuits are combined via colimit operations based on pushout and coproduct constructions in the category of control-driven circuits and their morphisms (cf. \cite{arellanes_compositional_2026}). A pushout is the colimit of a span of circuit morphisms which, according to Definition \ref{def:pushout}, has the purpose of effectively merging two circuits along a common one. We say it is common because it can be perceived as an interface that embodies the parts that are shared between the circuits being merged. More formally, each $A$-component of a pushout is the disjoint union of the $A$-component of one base circuit and the $A$-component of the other, with elements identified by the equivalence relation given by the corresponding span. 

\begin{definition}[Pushout]\label{def:pushout}
The pushout $\lambda_1+_{\lambda_0}\lambda_2$ of a span ${\rho:=\lambda_1\xleftarrow{\alpha}\lambda_0\xrightarrow{\beta}\lambda_2}$ of circuit morphisms is obtained by computing quotient sets componentwise. For example, the set of variables of $\lambda_1+_{\lambda_0}\lambda_2$ is given by $V_1 \sqcup V_2/\sim$ where $\sim$ is the finest equivalence relation given by ${\alpha_V(v)\sim\beta_V(v)}$ for all ${v\in V_0}$. The pushout of $\rho$ exists only if ${\alpha(\vec{i}(\beta))\cup\alpha(\vec{o}(\beta))\subseteq P_1^+\cup P_1^-}$ and ${\beta(\vec{i}(\alpha))\cup\beta(\vec{o}(\alpha))\subseteq P_2^+\cup P_2^-}$.
\end{definition}

The other colimit construction we consider is that of coproduct which basically serves to specify side-by-side circuits through monomorphisms that canonically ``inject'' circuits into a more complex one, as captured by Definition \ref{def:coproduct}. The universal properties of both coproduct and pushout follow from \cite{arellanes_compositional_2026}.

\begin{definition}[Coproduct]\label{def:coproduct}
The coproduct $\lambda_1+\lambda_2$ of circuits $\lambda_1$ and $\lambda_2$ is obtained by computing disjoint union componentwise such that the two canonical coproduct morphisms associated with $\lambda_1+\lambda_2$ are circuit monomorphisms.
\end{definition}

Below, we describe particular operators, built upon pushout and coproduct, for composing circuits into complex composites that encapsulate a precise control flow structure for evaluating concrete NAND operators in some specific order. 

\subsection{Sequential Circuits}

A \emph{sequential circuit} is constructed by pairing some/all the outvars of one circuit with some/all the invars of another via a sequentiable span of circuit morphisms. 

\begin{definition}[Sequentiable Span]\label{def:sequencing-span}
A span ${\lambda_1\xleftarrow{\alpha}\lambda_0\xrightarrow{\beta}\lambda_2}$ is sequentiable if $\lambda_0$ is a trivial circuit, and $\alpha$ and $\beta$ are monomorphisms satisfying ${\alpha_V(V_0)\subseteq V_1^-}$ and ${\beta_V(V_0)\subseteq V_2^+}$, respectively. If ${\alpha_V(V_0)=V_1^-}$ and ${\beta_V(V_0)=V_2^+}$, the span is totally sequentiable; otherwise, it is partially sequentiable.
\end{definition}

\begin{notation}\label{not:sequencing-span}
For Definition \ref{def:sequencing-span}, we call $\lambda_1$ and $\lambda_2$ the left and right operands, respectively. From now on, every time we write a sequentiable span, we will treat the bases of the left and right legs as left and right operands, correspondingly. 
\end{notation}

Definition \ref{def:sequencing-span} says that a sequentiable span requires a trivial circuit $\lambda_0$ at the apex, apart from two circuits ($\lambda_1$ and $\lambda_2$) at the base, with the restriction that (some or all) the invars of $\lambda_2$ must be identified with (some or all) the outvars of $\lambda_1$ injectively. This identification is not symmetric so, if outvars of $\lambda_2$ need to be identified with invars of $\lambda_1$, a different span would be required. In this case, $\lambda_1$ and $\lambda_2$ would correspondingly be the right and left operands, as per Notation \ref{not:sequencing-span}.

When all the outvars of $\lambda_1$ match all the invars of $\lambda_2$, the pushout of the corresponding sequentiable span is called a \emph{total sequential circuit}; otherwise, it is called a \emph{partial sequential circuit}. The corresponding formalisation is given in Definition \ref{def:sequential-computon} and a example of a total sequential composite is presented in Figure \ref{fig:example-basic}(c).

\begin{definition}[Sequential Circuit]\label{def:sequential-computon}
Let ${\rho:=\lambda_1\xleftarrow{\alpha}\lambda_0\xrightarrow{\beta}\lambda_2}$ be a sequentiable span of circuit morphisms. If $\rho$ is totally sequentiable, its pushout is called a total sequential circuit ${\lambda_1\unrhd_{\rho}\lambda_2}$. Otherwise, its pushout is a called a partial sequential circuit ${\lambda_1\rhd_{\rho}\lambda_2}$. The operations to form total and partial sequential circuits are called total sequencing and partial sequencing, respectively.
\end{definition}

Definition \ref{def:sequential-computon} says that, in our theory of control-driven Boolean circuits, we include a mechanism to sequence any two circuits regarding of the Boolean values they require or produce. This means that, even if the outvars of a circuit do not match all the invars of another (as in \emph{total sequencing}), they can still form a sequential circuit. This sort of composition, referred to as \emph{partial sequencing}, is possible since control variables always match because they have the same type (i.e., $\mathbbm{1}$). 

No matter whether a partial or a total sequential circuit is being formed, the morphisms induced by the corresponding pushout operation are mono (see Proposition \ref{prop:sequencing-mono}). This property follows from the fact that the components of every sequentiable span are monomorphisms too (see Definition \ref{def:sequencing-span}).

\begin{proposition}\label{prop:sequencing-mono}
If ${\lambda_1\xrightarrow{\gamma}\lambda_3\xleftarrow{\theta}\lambda_2}$ is the cospan induced by the pushout of a sequentiable span ${\lambda_1\xleftarrow{\alpha}\lambda_0\xrightarrow{\beta}\lambda_2}$, then $\gamma$ and $\theta$ are circuit monomorphisms.    
\end{proposition}
\begin{proof}
We prove only for $\gamma$, since the other is analogous. For this, we first note ${\gamma_\Sigma}$ is always injective because Definition \ref{def:morphism} says it is an inclusion function. For ${\gamma_I}$, $I_1$-elements are never identified with $I_2$-elements by the fact ${\lambda_0}$ is a trivial circuit with no flows at all, i.e., ${I_0=\emptyset}$. Thus, by Definition \ref{def:pushout}, ${\gamma_I}$ is necessarily injective since $I_3$ has an equivalence class for each $I_1$-element. The same property holds for ${\gamma_O}$ and ${\gamma_U}$ considering ${O_0=\emptyset=U_0}$.

Now, assume for contradiction that ${\gamma_V}$ is not injective so there are variables ${v_1,w_1\in V_1}$ with ${\gamma_V(v_1)=\gamma_V(w_1)}$. As ${V_3=V_1 \sqcup V_2/\sim}$, we have two possible scenarios: 
\begin{enumerate}
\item There is no $v_2\in V_2$ with ${\theta_V(v_2)=\gamma_V(v_1)=\gamma_V(w_1)}$. In this case, $v_1$ and $w_1$ form individual equivalence classes in $V_3$ for satisfying the reflexivity of $\sim$. Thus, contradicting our assumption that ${\gamma_V(v_1)=\gamma_V(w_1)}$. \label{prop:sequencing-mono-1}
\item There is some ${v_2\in V_2}$ where ${\theta_V(v_2)=\gamma_V(v_1)=\gamma_V(w_1)}$. In this case, the commutativity of pushouts entails there are ${v_0,w_0\in V_0}$ with ${\alpha_V(v_0)=v_1}$, ${\alpha_V(w_0)=w_1}$ and ${\beta_V(v_0)=v_2=\beta_V(w_0)}$. Clearly contradicting that ${\beta_V}$ is injective (see Definition \ref{def:sequencing-span}). \label{prop:sequencing-mono-2}
\end{enumerate}
As all their components are injective functions, we conclude $\gamma$ is a circuit monomorphism.\qed
\end{proof}

Intuitively, the interface of a sequential circuit preserves all the invars and all the outvars from the left and right operands, respectively (see Proposition \ref{prop:interface-sequential}). In the case of total sequencing, the invars of a total sequential circuit are identified with all the invars of the left operand, whereas the outvars are matched with all the outvars of the right operand (see Proposition \ref{prop:interface-sequential-total}).

\begin{proposition}\label{prop:interface-sequential}
Assume ${\rho:=\lambda_1\xleftarrow{\alpha}\lambda_0\xrightarrow{\beta}\lambda_2}$ is a sequentiable span of circuit morphisms. If ${\lambda_1\xrightarrow{\gamma}\lambda_3\xleftarrow{\theta}\lambda_2}$ is the cospan induced by the pushout of $\rho$, then ${\gamma_V(V_1^+)\subseteq V_3^+}$ and ${\theta_V(V_2^-)\subseteq V_3^-}$.
\end{proposition}
\begin{proof}
We just show ${\gamma_V(V_1^+)\subseteq V_3^+}$ since the proof of ${\theta_V(V_2^-)\subseteq V_3^-}$ is completely analogous. For this, assume for contrapositive ${v_3\notin V_3^+}$ so there is some ${o_3\in O_3}$ with ${t_3(o_3)=v_3}$. As ${O_3=O_1+_{O_0}O_2}$ according to Definition \ref{def:pushout}, we have three possible cases:
\begin{enumerate}
\item There exclusively is some ${o_1 \in O_1}$ with ${\gamma_O(o_1)=o_3}$ so, by the commutativity of $\gamma$, ${\gamma_V(t_1(o_1))=t_3(\gamma_O(o_1))=t_3(o_3)=v_3}$. As ${t_1(o_1)\notin V_1^+}$ as per Definition \ref{def:interface}, ${v_3\notin\gamma_V(V_1^+)}$.\label{prop:interface-sequential-1}
\item There exclusively is some ${o_2 \in O_2}$ with ${\theta_O(o_2)=o_3}$. In this case, the proof is analogous to that of \ref{prop:interface-sequential-1}.\label{prop:interface-sequential-2}
\item There are ${o_1\in O_1}$ and ${o_2\in O_2}$ for which ${\gamma_O(o_1)=o_3=\theta_O(o_2)}$. This case never holds since the apex of $\rho$ is a trivial circuit with no flows at all.\label{prop:interface-sequential-3}
\end{enumerate}
Proving \ref{prop:interface-sequential-3} does not hold and that ${v_3\notin V_3^+\implies v_3\notin\gamma_V(V_1^+)}$ for \ref{prop:interface-sequential-1} and \ref{prop:interface-sequential-2} entails ${\gamma_V(V_1^+)\subseteq V_3^+}$, as required.\qed
\end{proof}

\begin{proposition}\label{prop:interface-sequential-total}
Assume ${\rho:=\lambda_1\xleftarrow{\alpha}\lambda_0\xrightarrow{\beta}\lambda_2}$ is a totally sequentiable span of circuit morphisms. If ${\lambda_1\xrightarrow{\gamma}\lambda_3\xleftarrow{\theta}\lambda_2}$ is the cospan induced by the pushout of $\rho$, then ${\gamma_V(V_1^+)=V_3^+}$ and ${\theta_V(V_2^-)=V_3^-}$.
\end{proposition}
\begin{proof}
We only prove for ${\gamma_V(V_1^+)=V_3^+}$ since the other is completely analogous. For this, we let ${v_3\in V_3^+}$. Using Proposition \ref{prop:interface-preservation} and the fact ${V_3=V_1+_{V_0}V_2}$, we have three cases:
\begin{enumerate}
\item $v_3$ is exclusively identified with $\lambda_1$-invars so ${v_3\in\gamma_V(V_1^+)}$.\label{prop:interface-sequential-total-1}
\item $v_3$ is exclusively identified with $\lambda_2$-invars. Here, let ${v_2\in V_2^+}$ with ${\theta_V(v_2)=v_3}$. As $\rho$ is totally sequentiable, ${\beta_V(V_0)=V_2^+}$ so ${v_2\in V_2^+\iff v_2\in\beta_V(V_0)}$. Clearly, a contradiction to the fact $v_3$ is not identified with any $\lambda_1$-variable.\label{prop:interface-sequential-total-2}
\item there exist ${v_1\in V_1^+}$ and ${v_2\in V_2^+}$ where ${\gamma_V(v_1)=v_3=\theta_V(v_2)}$. As ${v_1\in V_1^+}$, ${v_3\in\gamma_V(V_1^+)}$ holds directly.\label{prop:interface-sequential-total-3}
\end{enumerate}
Disproving \ref{prop:interface-sequential-total-2} and proving ${v_3\in V_3^+\implies v_3\in\gamma_V(V_1^+)}$ for \ref{prop:interface-sequential-total-1} and \ref{prop:interface-sequential-total-3} implies that ${V_3^+\subseteq\gamma_V(V_1^+)}$. Thus, we simply use Proposition \ref{prop:interface-sequential} to conclude ${\gamma_V(V_1^+)=V_3^+}$.\qed
\end{proof}

Although sequencing is not commutative because the order of circuit execution matters, such a composition operation satisfies the identity property via the unit circuit (see Theorem \ref{th:sequencing-identity}). Theorem \ref{th:sequencing-associativity} says that only total sequencing satisfies associativity. Partial composition does not satisfy this property since the invars and outvars from the corresponding operands can remain unmatched, leading to non-isomorphic structures upon grouping.  

\begin{theorem}\label{th:sequencing-identity}
Up to isomorphism, the unit circuit $\Lambda$ serves as the left- and right-identity for (total or partial) sequencing.
\end{theorem}
\begin{proof}
We just prove for partial sequencing since such an operation is a generalisation of the other. For this, assume ${\Lambda\rhd_{\rho_1}\lambda_1}$ is the pushout of ${\Lambda \xleftarrow{\alpha} \lambda_0 \xrightarrow{\beta} \lambda_1}$. As the left operand is the unit circuit, it follows that ${\lambda_0}$ is essentially the same as $\Lambda$ so as to satisfy the mono property imposed by Definition \ref{def:sequencing-span}. Consequently, $\rho_1$ has the following shape:
\[
\begin{tikzcd}
\{\mathbbm{1}\} & \{v\} \arrow[l] & \emptyset\arrow[l]\arrow[r] & \emptyset & \emptyset\arrow[l]\arrow[r] & \{v\}\arrow[r] & \{\mathbbm{1}\} \\
\{\mathbbm{1}\} \arrow[u, hook, "\alpha_\Sigma"]\arrow[d, hook, "\beta_\Sigma"'] & \{v_0\}\arrow[l]\arrow[u, "\alpha_V"]\arrow[d, "\beta_V"'] & \emptyset\arrow[l]\arrow[r]\arrow[u, "\alpha_O"]\arrow[d, "\beta_O"'] & \emptyset\arrow[u, "\alpha_U"]\arrow[d, "\beta_U"'] & \emptyset\arrow[l]\arrow[r]\arrow[u, "\alpha_I"]\arrow[d, "\beta_I"'] & \{v_0\}\arrow[r]\arrow[u, "\alpha_V"]\arrow[d, "\beta_V"'] & \{\mathbbm{1}\}\arrow[u, hook, "\alpha_\Sigma"]\arrow[d, hook, "\beta_\Sigma"'] \\
\Sigma_1 & V_1 \arrow[l] & O_1\arrow[l]\arrow[r] & U_1 & I_1\arrow[l]\arrow[r] & V_1\arrow[r] & \Sigma_1
\end{tikzcd} 
\]
To show ${\Lambda\rhd_{\rho_1}\lambda_1\cong\lambda_1}$, it suffices to demonstrate ${\{v\}+_{\{v_0\}}V_1\cong V_1}$ which is trivially true by the fact all the variables of the left operand $\Lambda$ are identified in ${\Lambda\rhd_{\rho_1}\lambda_1}$. The proof for ${\lambda_1\rhd_{\rho_2}\Lambda\cong\lambda_1}$ is completely symmetric.\qed
\end{proof}

\begin{theorem}\label{th:sequencing-associativity}
Total sequencing is associative up to isomorphism.
\end{theorem}
\begin{proof}
If we assume that ${\lambda_1 \unrhd_{\rho_1}\lambda_2}$ and ${\lambda_2 \unrhd_{\rho_2}\lambda_3}$ are total sequential circuits with ${\rho_1:=\lambda_1 \xleftarrow{\alpha} \lambda_0 \xrightarrow{\beta} \lambda_2}$ and ${\rho_2:=\lambda_2 \xleftarrow{\phi} \lambda_4 \xrightarrow{\psi} \lambda_3}$, we have the following commutative diagram:
\[
\begin{tikzcd}
 & \lambda_4 \arrow[d, "\phi"']\arrow[r, "\psi"] & \lambda_3 \arrow[d, "\eta"] \\
\lambda_0 \arrow[d, "\alpha"']\arrow[r, "\beta"] & \lambda_2 \arrow[d, "\theta"]\arrow[r, "\zeta"] & \lambda_2 \unrhd_{\rho_2} \lambda_3 \arrow[d, "\mu"] \\
\lambda_1 \arrow[r, "\gamma"'] & \lambda_1 \unrhd_{\rho_1} \lambda_2 \arrow[r, "\kappa"'] & \lambda_5 
\end{tikzcd}
\]
where the cospans $(\gamma,\theta)$ and $(\zeta,\eta)$ are induced by the pushout of $\rho_1$ and $\rho_2$, respectively. Composing circuit morphisms horizontally and vertically, we obtain the following commutative diagrams:

\[
\begin{tikzcd}
\lambda_0 \arrow[d, "\alpha"']\arrow[r, "\zeta\circ\beta"] & \lambda_2 \unrhd_{\rho_2} \lambda_3 \arrow[d, "\mu"] & & \lambda_4 \arrow[d, "\theta\circ\phi"']\arrow[r, "\psi"] & \lambda_3 \arrow[d, "\mu\circ\eta"] \\
\lambda_1 \arrow[r, "\kappa\circ\gamma"'] & \lambda_5 & & \lambda_1 \unrhd_{\rho_1} \lambda_2 \arrow[r, "\kappa"'] & \lambda_5
\end{tikzcd}
\]

To show every $\lambda_0$-variable is mapped to an invar of $\lambda_2 \unrhd_{\rho_2} \lambda_3$, consider the following chain of equivalences: $v \in\zeta_V(\beta_V(V_0))\iff v \in \zeta_V(V_2^+)$ (because $\beta_V(V_0)=V_2^+$ by the fact that $\rho_1$ is totally sequentiable) $\iff v$ is an invar of $\lambda_2 \unrhd_{\rho_2} \lambda_3$ (as per Proposition \ref{prop:interface-sequential-total}).

Using Proposition \ref{prop:sequencing-mono} and the fact that $\rho_1$ is totally sequentiable, we have that $\beta$ and $\zeta$ are circuit monomorphisms. Consequently, ${\zeta\circ\beta}$ is also a monomorphism. Having ${\alpha_V(V_0)=V_1^-}$, because $\rho_1$ is totally sequentiable, we use Definition \ref{def:sequencing-span} to deduce ${\rho_3:=\lambda_1 \xleftarrow{\alpha} \lambda_0 \xrightarrow{\zeta\circ\beta} \lambda_2 \unrhd_{\rho_2} \lambda_3}$ is totally sequentiable. In this construction, $\lambda_0$ is the apex circuit whilst $\lambda_1$ and $\lambda_2\unrhd_{\rho_2}\lambda_3$ are the left and right operands, respectively. Using Definition \ref{def:sequential-computon}, we determine that the pushout of $\rho_3$ is the total sequential circuit ${\lambda_1 \unrhd_{\rho_3}(\lambda_2\unrhd_{\rho_2}\lambda_3)}$. A similar reasoning can be used to establish that ${(\lambda_1 \unrhd_{\rho_1}\lambda_2)\unrhd_{\rho_4}\lambda_3}$ is the pushout of the span ${\rho_4:=\lambda_1 \unrhd_{\rho_1}\lambda_2\xleftarrow{\theta\circ\phi} \lambda_4 \xrightarrow{\psi} \lambda_3}$. Therefore, ${\lambda_1 \unrhd_{\rho_3}(\lambda_2\unrhd_{\rho_2}\lambda_3)\cong\lambda_5\cong (\lambda_1 \unrhd_{\rho_1}\lambda_2)\unrhd_{\rho_4}\lambda_3}$, i.e., total sequencing is associative up to isomorphism.\qed
\end{proof}

\subsection{Parallel Circuits}

Asynchronous parallelising is a composition operation, which allows the formation of a \emph{parallel circuit} for the simultaneous execution of two circuits. Achieving this can directly be done by simply computing a coproduct (see Definition \ref{def:pasync}). 

\begin{definition}\label{def:pasync}
A parallel circuit ${\lambda_1+\lambda_2}$ is a coproduct of $\lambda_1$ and $\lambda_2$.
\end{definition}

As coproduct is trivially commutative and associative, the operation to form a parallel circuit satisfies these properties (see Theorem \ref{th:parallelising-associativity-commutativity}). Identity is not met since this would require the existence of empty circuits which, by Definition \ref{def:computon}, is impossible given that every circuit must have at least one variable. 

\begin{theorem}\label{th:parallelising-associativity-commutativity}
Parallelising is associative and commutative up to isomorphism.
\end{theorem}
\begin{proof}
These two properties follow immediately from the fact that coproduct is built upon disjoint union of finite sets which is a well-known operation that satisfies both associativity and commutativity.\qed
\end{proof}

\subsection{Branching Circuits}

A \emph{branching circuit} is a composite that identifies invars with invars and outvars with outvars from two circuits, meaning it can only be formed from circuits with isomorphic in-adjoint domains and isomorphic out-adjoint domains (see Definition \ref{def:branching}). As a consequence of this constraint, the invars and outvars of a branching circuit always match the invars and outvars of the circuits being branched over (see Proposition \ref{prop:adjoints-for-branching}). The operation to form a branching composite can always be computed in the category of circuits by relying on the fact that the pushout of a span of adjoint morphisms can always be computed (cf. \cite{arellanes_compositional_2026}). Consequently, isomorphisms reside in such a category.

\begin{definition}[Branching Circuit]\label{def:branching}
Let $\rho$ be the diagram formed by the spans ${\lambda_2\xleftarrow{\lambda_2^+}\lambda_0\xrightarrow{\lambda_3^+}\lambda_3}$ and ${\lambda_2\xleftarrow{\lambda_2^-}\lambda_1\xrightarrow{\lambda_3^-}\lambda_3}$ of adjoint morphisms. A branching circuit $\lambda_2?_{\rho}\lambda_3$ is the colimit of $\rho$, computed as $\lambda_2+_{\lambda_0+\lambda_1}\lambda_3$. By Definition \ref{def:adjoints}, $\lambda_0$ and $\lambda_1$ are necessarily trivial circuits.
\end{definition}

\begin{proposition}\label{prop:adjoints-for-branching}
If $\rho$ is the diagram formed by the spans ${\lambda_2\xleftarrow{\lambda_2^+}\lambda_0\xrightarrow{\lambda_3^+}\lambda_3}$ and ${\lambda_2\xleftarrow{\lambda_2^-}\lambda_1\xrightarrow{\lambda_3^-}\lambda_3}$, $\lambda_0$ and $\lambda_1$ are the domains of the in- and out-adjoints of ${\lambda_2?_{\rho}\lambda_3}$, respectively.
\end{proposition}
\begin{proof}
Assuming ${\alpha\colon\lambda_0\rightarrow\lambda_0+\lambda_1}$ and ${\beta\colon\lambda_1\rightarrow\lambda_0+\lambda_1}$ are the canonical injections induced by the coproduct ${\lambda_0+\lambda_1}$, we obtain the following diagram:
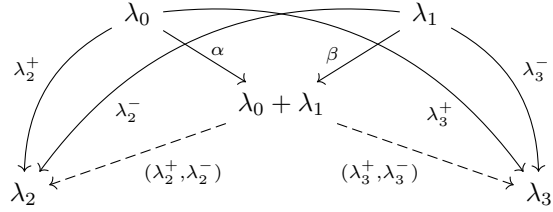
\[
\begin{tikzcd}
 & \lambda_0\arrow[ddl,bend right=30,"\lambda_2^+"']\arrow[ddrrr,bend left=30,"\lambda_3^+"{xshift=2em,yshift=-3.5em}]\arrow[dr,"\alpha"{yshift=-0.2em}] & & \lambda_1\arrow[ddr,bend left=30,"\lambda_3^-"]\arrow[ddlll,bend right=30,"\lambda_2^-"{xshift=-3.3em,yshift=-2em}]\arrow[dl,"\beta"'{xshift=-0.5em,yshift=-0.5em}] & \\
 & & \lambda_0+\lambda_1\arrow[dll,dash pattern=on 4pt off 2pt,"{(\lambda_2^+,\lambda_2^-)}"]\arrow[drr,dash pattern=on 4pt off 2pt,"{(\lambda_3^+,\lambda_3^-)}"'] & & \\
 \lambda_2 & & & & \lambda_3 
\end{tikzcd}
\]
where ${(\lambda_2^+,\lambda_2^-)}$ and ${(\lambda_3^+,\lambda_3^-)}$ are the unique morphisms derived from the universal property of coproducts. This property follows from the fact that coproduct is computed in the category of finite sets and total functions. 

Now, by Definition \ref{def:branching}, we know that the colimit ${\lambda_4:=\lambda_2?_{\rho}\lambda_3}$ of the above diagram corresponds to the pushout of ${(\lambda_2^+,\lambda_2^-)}$ and ${(\lambda_3^+,\lambda_3^-)}$. If ${\gamma\colon\lambda_2\rightarrow\lambda_4}$ and ${\theta\colon\lambda_3\rightarrow\lambda_4}$ are the morphisms induced by such a pushout, the following equalities must hold by commutativity: 

\[
\gamma\circ\lambda_2^+=\gamma\circ(\lambda_2^+,\lambda_2^-)\circ\alpha=\theta\circ(\lambda_3^+,\lambda_3^-)\circ\alpha=\theta\circ\lambda_3^+
\]
\[
\theta\circ\lambda_3^-=\gamma\circ(\lambda_2^+,\lambda_2^-)\circ\beta=\theta\circ(\lambda_3^+,\lambda_3^-)\circ\beta=\gamma\circ\lambda_2^-
\]

With the above in mind, we now prove as follows: ${v_4\in V_4^+ \iff }$ there is some ${v_2\in V_2^+}$ and some ${v_3\in V_3^+}$ where ${\gamma_V(v_2)=v_4=\theta_V(v_3)}$ (see Proposition \ref{prop:interface-preservation}) $\iff$ there is some ${v\in V_0+V_1}$ where ${(\lambda_2^+,\lambda_2^-)(v)=v_2}$ and ${(\lambda_3^+,\lambda_3^-)(v)=v_3}$ (as per the commutativity property of pushouts). As ${v_2\in V_2^+}$ and ${v_3\in V_3^+}$, Definition \ref{def:adjoints} entails that commutativity only holds if ${v_2=\lambda_2^+(v_0)}$ and ${v_3=\lambda_3^+(v_0)}$ for some ${v_0\in V_0}$. Using the above equations, we have:
\begin{small}
\[
(\gamma_V\circ\lambda_2^+)(v_0)=(\gamma_V\circ(\lambda_2^+,\lambda_2^-)\circ\alpha_V)(v_0)=(\theta_V\circ(\lambda_3^+,\lambda_3^-)\circ\alpha_V)(v_0)=(\theta_V\circ\lambda_3^+)(v_0)=v_4
\]
\end{small}

Thus, proving that the morphisms ${\gamma\circ\lambda_2^+\colon\lambda_0\rightarrow\lambda_4}$, ${\gamma\circ(\lambda_2^+,\lambda_2^-)\circ\alpha\colon\lambda_0\rightarrow\lambda_4}$, ${\theta\circ(\lambda_3^+,\lambda_3^-)\circ\alpha\colon\lambda_0\rightarrow\lambda_4}$ and ${\theta\circ\lambda_3^+\colon\lambda_0\rightarrow\lambda_4}$ are in-adjoints of ${\lambda_4}$. A similar reasoning can be used to show that ${\theta\circ\lambda_3^-\colon\lambda_1\rightarrow\lambda_4}$, ${\gamma\circ(\lambda_2^+,\lambda_2^-)\circ\beta\colon\lambda_1\rightarrow\lambda_4}$, ${\theta\circ(\lambda_3^+,\lambda_3^-)\circ\beta\colon\lambda_1\rightarrow\lambda_4}$ and ${\gamma\circ\lambda_2^-\colon\lambda_1\rightarrow\lambda_4}$ are out-adjoints of $\lambda_4$.\qed
\end{proof}

Although branching is both commutative and associative (by Theorems \ref{th:branching-commutativity} and \ref{th:branching-associativity}), the identity law is not satisfied in general since that would require a trivial circuit identified with both invars and outvars; thus, potentially violating the constraint that every circuit must have invars and outvars (see Definition \ref{def:computon}). 

\begin{theorem}\label{th:branching-commutativity}
Branching is commutative up to isomorphism.
\end{theorem}
\begin{proof}
If ${\lambda_1?_{\rho_1}\lambda_2}$ and ${\lambda_2?_{\rho_2}\lambda_1}$ are two branching circuits, Definition \ref{def:branching} says each of them must be the pushout of circuit morphisms. As these morphisms are built upon total functions between finite sets (see Definition \ref{def:morphism}), it is clear that pushout is computed in the category of finite sets and total functions. In this category, it is well-known that pushout is commutative up to unique isomorphism. Hence, ${\lambda_1?_{\rho_1}\lambda_2\cong\lambda_2?_{\rho_2}\lambda_1}$ must hold trivially.\qed
\end{proof}

\begin{theorem}\label{th:branching-associativity}
Branching is associative up to isomorphism.
\end{theorem}
\begin{proof}
Assume ${\lambda_2?_{\rho_1}\lambda_3}$ and ${\lambda_3?_{\rho_2}\lambda_4}$ are branching circuits where:
\begin{itemize}
\item $\rho_1$ is the diagram formed by the spans ${\lambda_2\xleftarrow{\lambda_2^+}\lambda_0\xrightarrow{\lambda_3^+}\lambda_3}$ and ${\lambda_2\xleftarrow{\lambda_2^-}\lambda_1\xrightarrow{\lambda_3^-}\lambda_3}$, and
\item $\rho_2$ is the diagram formed by the spans ${\lambda_3\xleftarrow{\lambda_3^+}\lambda_0\xrightarrow{\lambda_4^+}\lambda_4}$ and ${\lambda_3\xleftarrow{\lambda_3^-}\lambda_1\xrightarrow{\lambda_4^-}\lambda_4}$.
\end{itemize}
By Definition \ref{def:branching}, ${\lambda_2?_{\rho_1}\lambda_3}$ and ${\lambda_3?_{\rho_2}\lambda_4}$ are constructed through the operations ${\lambda_2+_{\lambda_0+\lambda_1}\lambda_3}$ and ${\lambda_3+_{\lambda_0+\lambda_1}\lambda_4}$, respectively. Accordingly, we have the following commutative diagram:

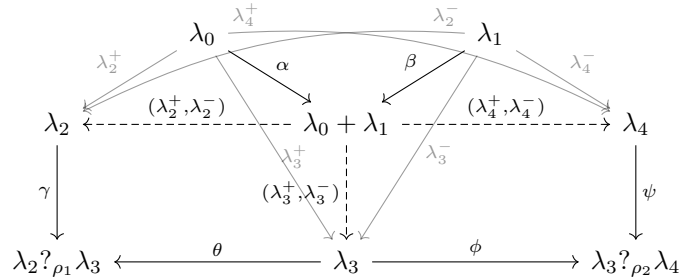
\begin{center}
\begin{tikzcd}
 & \lambda_0 \arrow[dr, "\alpha"]\arrow[dl, opacity=0.4, "\lambda_2^+"']\arrow[drrr, bend left=15, opacity=0.4,"\lambda_4^+"{yshift=4pt,xshift=-75pt}]\arrow[dddr, opacity=0.4, "\lambda_3^+"{yshift=-10pt,xshift=0pt}] & & \lambda_1 \arrow[dl, "\beta"']\arrow[dr, opacity=0.4, "\lambda_4^-"]\arrow[dlll, bend right=15, opacity=0.4,"\lambda_2^-"'{yshift=4pt,xshift=75pt}]\arrow[dddl, opacity=0.4, "\lambda_3^-"{yshift=5pt}] & \\
 \lambda_2 \arrow[dd, "\gamma"'] & & \lambda_0+\lambda_1 \arrow[ll, dashed, "{(\lambda_2^+,\lambda_2^-)}"']\arrow[rr, dashed, "{(\lambda_4^+,\lambda_4^-)}"]\arrow[dd, dashed, "{(\lambda_3^+,\lambda_3^-)}"'] & & \lambda_4 \arrow[dd, "\psi"] \\
 & & & & \\
\lambda_2?_{\rho_1}\lambda_3 & & \lambda_3 \arrow[ll, "\theta"']\arrow[rr, "\phi"] & & \lambda_3?_{\rho_2}\lambda_4
\end{tikzcd}
\end{center}

Here, $\alpha$ and $\beta$ are the canonical injections into ${\lambda_0+\lambda_1}$, and the cospan ${(\gamma,\theta)}$ is induced by the pushout of ${(\lambda_2^+,\lambda_2^-)}$ and ${(\lambda_3^+,\lambda_3^-)}$ which, in turn, are unique morphisms induced by the universal property of coproducts. Similarly, the cospan ${(\phi,\psi)}$ derives from the pushout of ${(\lambda_3^+,\lambda_3^-)}$ and ${(\lambda_4^+,\lambda_4^-)}$ which are also induced from the universal property of coproducts. 

Now, by the commutativity property of pushouts, we deduce the existence of the morphisms ${f\colon\lambda_0+\lambda_1\rightarrow\lambda_2?_{\rho_1}\lambda_3}$ and ${g\colon\lambda_0+\lambda_1\rightarrow\lambda_3?_{\rho_2}\lambda_4}$ which are given by ${f=\gamma\circ(\lambda_2^+,\lambda_2^-)=\theta\circ(\lambda_3^+,\lambda_3^-)}$ and ${g=\phi\circ(\lambda_3^+,\lambda_3^-)=\psi\circ(\lambda_4^+,\lambda_4^-)}$, respectively. Thus, yielding the following diagrams:\\

\noindent
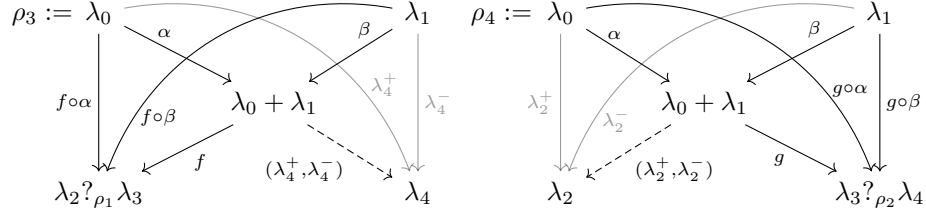
\begin{minipage}{0.45\textwidth}
\centering
\begin{tikzcd}
 \rho_3:= & [-1.4cm] \lambda_0 \arrow[dr, "\alpha"{yshift=4pt,xshift=-10pt}]\arrow[dd, "f\circ\alpha"']\arrow[ddrr, opacity=0.4, bend left=40, "\lambda_4^+"{yshift=-25pt,xshift=25pt}] & & \lambda_1 \arrow[dl, "\beta"'{yshift=4pt,xshift=10pt}]\arrow[dd, opacity=0.4, "\lambda_4^-"]\arrow[ddll, bend right=40, "f\circ\beta"{yshift=-25pt,xshift=-32pt}] \\
 & & \lambda_0+\lambda_1 \arrow[dl,"f"]\arrow[dr, dashed,"{(\lambda_4^+,\lambda_4^-)}"'] & \\
 & \lambda_2?_{\rho_1}\lambda_3 & & \lambda_4
\end{tikzcd}
\end{minipage}
\hfill
\begin{minipage}{0.6\textwidth}
\centering
\begin{tikzcd}
 \rho_4:= & [-0.95cm] \lambda_0 \arrow[dr, "\alpha"{yshift=4pt,xshift=-10pt}]\arrow[dd, opacity=0.4, "\lambda_2^+"']\arrow[ddrr, bend left=40, "g\circ\alpha"{yshift=-25pt,xshift=25pt}] & & \lambda_1 \arrow[dl, "\beta"'{yshift=4pt,xshift=10pt}]\arrow[dd, "g\circ\beta"]\arrow[ddll, bend right=40, opacity=0.4, "\lambda_2^-"{yshift=-25pt,xshift=-32pt}] \\
 & & \lambda_0+\lambda_1 \arrow[dl, dashed, "{(\lambda_2^+,\lambda_2^-)}"]\arrow[dr, "g"'] & \\
 & \lambda_2 & & \lambda_3?_{\rho_2}\lambda_4
\end{tikzcd}
\end{minipage}
\\

For the diagram $\rho_3$, we use Proposition \ref{prop:adjoints-for-branching} to deduce that ${f\circ\alpha\colon\lambda_0\rightarrow\lambda_2?_{\rho_1}\lambda_3}$ and ${f\circ\beta\colon\lambda_1\rightarrow\lambda_2?_{\rho_1}\lambda_3}$ are the in- and out-adjoints of ${\lambda_2?_{\rho_1}\lambda_3}$, respectively. Simply applying Definition \ref{def:branching}, we determine ${(\lambda_2?_{\rho_1}\lambda_3)?_{\rho_3}\lambda_4}$ is the pushout of $f$ and ${(\lambda_4^+,\lambda_4^-)}$, computed as ${(\lambda_2+_{\lambda_0+\lambda_1}\lambda_3)+_{\lambda_0+\lambda_1}\lambda_4}$. A similar reasoning allow us to deduce ${\lambda_2?_{\rho_4}(\lambda_3?_{\rho_2}\lambda_4)}$ is the pushout of ${(\lambda_2^+,\lambda_2^-)}$ and $g$, computed as ${\lambda_2+_{\lambda_0+\lambda_1}(\lambda_3+_{\lambda_0+\lambda_1}\lambda_4)}$.

As ${(\lambda_j^+,\lambda_j^-)\colon\lambda_0+\lambda_1\rightarrow\lambda_j}$ is a monorphism for all ${j\in\{2,3,4\}}$ (by Definition \ref{def:adjoints}) and pushouts are computed in the category of finite sets, associativity of pushouts holds up to unique isomorphism. That is,  ${(\lambda_2+_{\lambda_0+\lambda_1}\lambda_3)+_{\lambda_0+\lambda_1}\lambda_4}\cong{\lambda_2+_{\lambda_0+\lambda_1}(\lambda_3+_{\lambda_0+\lambda_1}\lambda_4)}$. Or, equivalently, ${(\lambda_2?_{\rho_1}\lambda_3)?_{\rho_3}\lambda_4}\cong{\lambda_2?_{\rho_4}(\lambda_3?_{\rho_2}\lambda_4)}$. Thus, concluding that branching is associative up to unique isomorphism.\qed
\end{proof}

\subsection{Iterative Circuits}

In the theory of control-driven Boolean circuits, there are \emph{head-} and \emph{tail-iterative} circuits. In the former, the decision of continuing an iteration is made just before executing the circuit being iterated over (as in \emph{while} programming constructs), whereas in the latter such a decision is made immediately after executing the circuit (as in \emph{do-while} constructs). In any case, the construction requires the formation of a diagram of six adjoint morphisms from which a colimit is computed. The adjoints sharing domain identify the variables being glued together.

Definition \ref{def:head} describes the colimit operation to form a head-iterative circuit, whereas Definition \ref{def:tail} shows how to construct a tail-iterative one. Both operations can always be computed in the category of circuits by relying on the fact that the pushout of a span of adjoint morphisms can always be computed (cf. \cite{arellanes_compositional_2026}).

\begin{definition}[Head-Iterative Circuit]\label{def:head}
A head-iterative circuit $*_\rho\lambda$ is the colimit of a diagram $\rho$ of adjoint morphisms:
\[
\begin{tikzcd}
\lambda_2 & \lambda_4 & \lambda_0 \arrow[l, "\lambda_4^+"]\arrow[r, "\lambda^+"']\arrow[ll, bend right=12, "\lambda_2^-"']\arrow[rrr, bend left=12, "\lambda_3^-"] & \lambda & \lambda_1 \arrow[l, "\lambda^-"]\arrow[r, "\lambda_3^+"'] & \lambda_3
\end{tikzcd}
\]
where $\lambda$, $\lambda_2$, $\lambda_3$ and $\lambda_4$ are sound circuits, and $\lambda_0$ and $\lambda_1$ are trivial (see Definition \ref{def:adjoints}). The colimit of $\rho$ is computed as ${(\lambda_4+_{\lambda_0}\lambda)+_{(\lambda_0+\lambda_1)}(\lambda_2+_{\lambda_0}\lambda_3)}$.
\end{definition}

\begin{definition}[Tail-Iterative Circuit]\label{def:tail}
A tail-iterative circuit $\lambda*_\rho$ is the colimit of a diagram $\rho$ of adjoint morphisms:
\[
\begin{tikzcd}[row sep=0.9em]
 & \lambda_0 \arrow[dl, "\lambda_2^-"']\arrow[r, "\lambda_3^-"]\arrow[dr, "\lambda^+"'] & \lambda_3 & \lambda_1 \arrow[l, "\lambda_3^+"']\arrow[dr, "\lambda_4^+"]\arrow[dl, "\lambda^-"] \\
 [-0.5cm]
\lambda_2 & & \lambda & & \lambda_4
\end{tikzcd}
\]
where $\lambda$, $\lambda_2$, $\lambda_3$ and $\lambda_4$ are sound circuits, and $\lambda_0$ and $\lambda_1$ are trivial (see Definition \ref{def:adjoints}). The colimit of $\rho$ is computed as ${((\lambda_2+_{\lambda_0}\lambda_3)+_{\lambda_3}(\lambda_3+_{\lambda_1}\lambda_4))+_{(\lambda_0+\lambda_1)}\lambda}$.
\end{definition}

In both cases, $\lambda_2$ and $\lambda_4$ respectively serve as entry and exit points for the iterative structure, whereas $\lambda$ and $\lambda_3$ correspond to the circuit being iterated over and the circuit marking the end of an iteration, respectively. Introducing additional circuits other than $\lambda$ is necessary since these entities always need interface variables to interact with others, given the restrictions imposed by Definition \ref{def:morphism}. Definitions \ref{def:head} and \ref{def:tail} just differ in how adjoint morphisms share domains so as to structurally direct the position of the point that decides loop termination.

\subsection{Circuit Dynamics}
\label{sec:dynamics}

A circuit defines an explicit control flow structure for the activation of computation units in some order, no matter whether it is sound or not. Given that a unit is a family of operators, activating it means reducing it by choosing an operator for evaluation at the values from the variables connected to that unit. Thus, variables serve as placeholders for information, ranging from control signals to Boolean values. Boolean values are optional but control signals are ever present because control variables always exist in any circuit interface (see Definition \ref{def:computon}). 

The state of a circuit is a mapping from variables to concrete values at some moment in discrete time, which is initial if it assigns a value to each invar and no value to all the other variables. A circuit reaches its final state when each outvar has a value, while all the other variables have no value at all. Value absence is expressed through the use of partial functions, as described in Definition \ref{def:computon-state}.

\begin{definition}[Circuit State] \label{def:computon-state}
The state of a circuit $\lambda$ at time $j$ is a partial function ${\delta^j\colon V\rightarrow\{\ast,0,1\}}$ where for all ${v \in Dom(\delta^j)}$: 
\begin{itemize}
\item[] ${c(v)=\mathbbm{1}\iff\delta^j(v)\in\{\ast\}}$ and
\item[] ${c(v)=\mathbb{B}\iff\delta^j(v)\in\{0,1\}}$
\end{itemize}
A state ${\delta^j}$ is initial if ${Dom(\delta^j)=V^+}$ or final if ${Dom(\delta^j)=V^-}$.
\end{definition}

\begin{remark}
For convention and since we deal with discrete time, the initial state of every circuit occurs at time $0$. We use $*$ to express a control signal.
\end{remark}

At every step of a circuit execution, a number of computation units can be enabled. A unit is enabled at time $j$ when all the variables connected to it have values assigned by the current state at $j$ (see Definition \ref{def:computation-unit-status}). Computation units not satisfying this condition are called idle.

\begin{definition}[Unit Status] \label{def:computation-unit-status}
Given a circuit $\lambda$ and a state ${\delta^j}$, a computation unit ${u\in U}$ is enabled under ${\delta^j}$ if ${\bullet u\subseteq Dom(\delta^j)}$; otherwise, $u$ is idle under ${\delta^j}$. 
\end{definition}

As control-driven Boolean circuits can behave non-deterministically (e.g., branching circuits), it is necessary to decide which computation units can be reduced at every time step. Definition \ref{def:computation-unit-ready-reduction} says that this is done by partitioning the set of enabled computation units under $\delta^j$, by taking into account matching invar interfaces. And then invoking the axiom of choice to non-deterministically choose one representative from each equivalence class, in order to form a set $R^j$. Evidently, when a partition contains only one unit $u$, the choice function will select $u$, i.e., $u$ will be ready to be reduced.\footnote{A random choice function defines arbitrary selection rather than probabilistic randomness. In practice, such a choice could be implemented via stochastic methods.}

\begin{definition}[Ready Units] \label{def:computation-unit-ready-reduction}
Let ${E^j}$ be the finite set of computation units enabled under ${\delta^j}$ and $\sim$ the equivalence relation ${\{(u_1,u_2)\in E^j\times E^j\mid\bullet u_1=\bullet u_2\}}$. If $A$ is the partition induced by $\sim$ and $f$ is the random choice function on $A$, the set ${R^j}$ of computation units ready to be reduced under ${\delta^j}$ is ${\{f(E)\mid E\in A\}}$. 
\end{definition}

Strictly speaking, computation units in $R^j$ are reduced at ${j+1}$ through the mapping given by $\delta^j$; thereby, yielding new values in the variables connected from such units (see Definition \ref{def:expression-evaluation}).\footnote{To understand Definition \ref{def:expression-evaluation}, recall that a computation unit is a function that establishes a family of NAND operators indexed by $\mathbb{N}$ (see Definition \ref{def:cunit}).} These new values can subsequently be consumed by units attached from such variables. The overall execution of a circuit is given by a state transition function whose behaviour is captured in Definition \ref{def:transition}.

\begin{definition}[Unit Reduction] \label{def:expression-evaluation}
If $\lambda$ is a circuit with ${u\in R^j}$ and $B$ is the set ${\{v\in\bullet u\mid c(v)=\mathbb{B}\}}$, the result ${\llbracket u \rrbracket^j}$ of reducing $u$ under ${\delta^j}$ is given by:

\begin{align*}
  \llbracket u \rrbracket^j &= u(|B|)(\delta^j(b_1),\ldots,\delta^j(b_{|B|}))=\uparrow_{|B|}(\delta^j(b_1),\ldots,\delta^j(b_{|B|})) & \\
    &=\begin{cases}
       1 & |B|=0\\
       \lnot(\delta^j(b_1)\land\cdots\land \delta^j(b_{|B|})) & \text{otherwise}
    \end{cases} &\!\!\!\!\!\!\!\!\! \text{ where }{\bigcup_{k=1}^{|B|} \{b_k\}=B}
\end{align*}
\end{definition}

\begin{remark}\label{rem:nand-behaviour}
For evaluating a particular operator from a family, we use prefix notation with parentheses included. In the case of ${|B|=0}$, the arguments are ignored because the constant function ${\uparrow_0}$ is chosen. When ${|B|=1}$, ${\uparrow_1}$ evidently behaves as the NOT function. When ${|B|>1}$, the ${|B|}$-ary NAND operator is chosen in which case the order of arguments is irrelevant, given that NAND is based on AND (i.e., an operator that satisfies commutativity). 
\end{remark}

\begin{definition}[State Transition]\label{def:transition}
Given the state ${\delta^j}$ of a circuit $\lambda$ for ${j\geq 0}$, ${\delta^{j+1}}$ is computed as follows for each ${v\in V}$:
\[\delta^{j+1}(v) = 
    \begin{cases}
       \ast & (\exists u \in R^j)[v\in u\bullet\text{ and }c(v)=\mathbbm{1}]\\
       \llbracket u \rrbracket^j & (\exists u \in R^j)[v\in u\bullet\text{ and }c(v)=\mathbb{B}]\\
       \delta^j(v) & v\in Dom(\delta^j)\text{ and }(\nexists u \in R^j)[v\in\bullet u\cup u\bullet]\\ 
    \end{cases}
\]
\end{definition}

\begin{remark}\label{rem:domain-state}
Definitions \ref{def:computon-state} and \ref{def:transition} imply ${V^+}$ is the domain of the initial state ${\delta^0}$ of a circuit $\lambda$ and that ${Dom(\delta^{j+1})}$ is $\{v\in V\mid (\exists u \in R^{j})[v\in u\bullet]\}\cup\{v\in Dom(\delta^{j})\mid (\nexists u \in R^{j})[v\in\bullet u\cup u\bullet]\}$ for ${j\geq 0}$. This means that there is a halting recursive procedure to determine the domain of an arbitrary state, i.e., it is possible to establish which variables have values at every step of a circuit execution.
\end{remark}

Based on Definition \ref{def:transition} and considering that control variables are ever present in any circuit $\lambda$, there are four possible execution patterns for any unit $u\in U$:
\begin{itemize}
    \item If ${c(\bullet u)=\{\mathbbm{1}\}=c(u\bullet)}$, all variables in $u\bullet$ take the value $*$ because there are no Boolean variables. This behaviour can be suitable to replicate (or merge) control signals onto (or from) multiple circuits.
    \item If ${c(\bullet u)=\{\mathbbm{1},\mathbb{B}\}}$ and ${c(u\bullet)=\{\mathbbm{1}\}}$, $u$ produces control signals only. So, $u$ can be perceived as a computation unit that discards (or ``eats'') Booleans.
    \item If ${c(\bullet u)=\{\mathbbm{1}\}}$ and ${c(u\bullet)=\{\mathbbm{1},\mathbb{B}\}}$, the operator $\uparrow_0$ is triggered for $u$ to produce control signals and the constant $1$ in all the Boolean variables in ${u\bullet}$. 
    \item If ${c(\bullet u)=\{\mathbbm{1},\mathbb{B}\}=c(u\bullet)}$, a classical NAND gate is simulated. By Remark \ref{rem:nand-behaviour}, when there is exactly one Boolean variable in ${\bullet u}$, all the Boolean variables in ${u\bullet}$ take the negation of the only input since $\uparrow_1$ is chosen. Otherwise, all the Boolean variables in ${u\bullet}$ take the value resulting from the action of the corresponding $k$-ary NAND operator (${k>1}$). Evidently, all the control variables in ${u\bullet}$ take the value $*$.
\end{itemize}

The above cases are just a description of the possible values a computation unit can produce during a state transition, depending on the unit's structure. Such transitions describe how to move from an initial state to a final one, as described in Definition \ref{def:termination}. Below this definition we describe a simple example.

\begin{definition}[Circuit Termination]\label{def:termination}
A circuit $\lambda$ terminates if and only if there is a finite orbit of states from its initial state to its final state.
\end{definition}

\begin{example}
Let ${\delta^0=\{(v_1,*),(v_2,1),(v_3,0)\}}$ be the initial state of the AND circuit from Figure \ref{fig:example-basic}(c). As ${v_1,v_2,v_3\in\bullet\uparrow_2}$, such a unit is enabled so ${E^0=\{\uparrow_2\}}$. Having only one element in $E^0$ means that the partition induced by the equivalence relation from Definition \ref{def:computation-unit-ready-reduction} must be ${\{\{\uparrow_2\}\}}$ so ${R^0=\{\uparrow_2\}}$. Using Definition \ref{def:transition}, we obtain 
\[\delta^1=\{(v_4,*),(v_5,\lnot(\delta^0(v_2)\land\delta^0(v_3)))\}=\{(v_4,*),(v_5,1)\}\] 
which enables and triggers the reduction of $\uparrow_1$ to yield
\[\delta^2=\{(v_6,*),(v_7,\lnot\delta^1(v_5))\}=\{(v_6,*),(v_7,0)\}\] 
As ${Dom(\delta^2)=\{v_6,v_7\}=V^-}$, ${\delta^2}$ is the final state of the composite shown in Figure \ref{fig:example-basic}(c). That is, our composite encapsulates sequential control flow for the computation of ${\lnot(\lnot(\delta^0(v_2)\land\delta^0(v_3)))}$ in two time steps.
\end{example}

\section{Relationship Between Control-Driven Boolean Circuits and Classical Boolean Circuits}
\label{sec:equivalence}

Boolean circuits are a model of computation for combining logic gates into distinct configurations in order to perform some Boolean function. In this model, a gate is the fundamental unit of computation which implements a concrete Boolean function from a pre-defined basis set \cite{clote_boolean_2002}. Here, we consider $\{\uparrow\}$ since it is well-known that such a set forms a functionally complete basis from which all the other gates can be constructed \cite{wegener_complexity_1987}. Definition \ref{def:classical-circuit} formalises the notion of Boolean circuits over this basis, which we simply refer to as \emph{NAND circuits}. 

\begin{definition}[NAND Circuit]\label{def:classical-circuit}
A NAND circuit $C$ is a directed acyclic graph ${(N,E)}$ where $N$ is a non-empty set of vertices such that each ${n\in N}$ satisfies ${n^+\neq 0\lor n^-\neq 0}$ and corresponds to an input variable (if ${n^+=0}$), an output variable (if ${n^-=0}$) or a NAND gate (if ${n^+=2}$ and ${n^-=1}$). The set $E$ of edges is non-empty and each ${e\in E}$ denotes flow of information from an input variable to a gate, from a gate to another gate or from a gate to an output variable. 
\end{definition}

\begin{remark}\label{rem:circuits-always-gates}
The condition ${n^+\neq 0\lor n^-\neq 0}$ means there are no isolated variables, which is a valid assumption as they can be expressed as buffers built out of NAND gates. Due to the way we define edges, gates are necessarily present in $N$. 
\end{remark}

\begin{remark}
Having fan-in-$2$ fan-out-$1$ gates does not limit computation power (i.e., what can be computed), since $k$-ary gates can be built out of them. Although this structural feature increases circuit size, it simplifies the transformation from NAND to control-driven circuits. Our transformation procedure is described in the proof of Theorem \ref{th:classical-to-control} which, by Proposition \ref{prop:sound-transformation}, always yields a sound circuit.
\end{remark}

\begin{theorem}\label{th:classical-to-control}
Every NAND circuit is a control-driven Boolean circuit.    
\end{theorem}
\begin{proof}
Given a NAND circuit $C$, let:
\begin{itemize}[noitemsep, topsep=4pt]
\item ${E_1:=E\setminus\{(m,n)\in E\mid m^+=0\land n^+=2\land n^-=1\}}$, 
\item ${E_2:=E\setminus\{(m,n)\in E\mid m^+=2\land m^-=1\land n^-=0\}}$,
\item ${G:=\{n\in N\mid n^+=2\land n^-=1\}}$,
\item ${s'\colon E_1\sqcup E_1\rightarrow G}$ be a function given by ${s'((m,n),j)=m}$ where ${(m,n)\in E_1}$ and ${j\in\{1,2\}}$, and ${t'\colon E_2\sqcup E_2\rightarrow G}$ be a function given by ${t'((x,y),k)=y}$ where ${(x,y)\in E_2}$ and ${k\in\{1,2\}}$.
\end{itemize}
First, notice that $E_1$ discards all the edges that run from an input variable to a gate; so, by Definition \ref{def:classical-circuit}, each ${(m,n)\in E_1}$ is an edge from a gate to an output variable or an edge between gates. Similarly, $E_2$ discards all the edges running from a gate to an output variable; so, again by Definition \ref{def:classical-circuit}, each ${(m,n)\in E_2}$ is an edge where $n$ is a gate and $m$ is an input variable or another gate. The set $G$ simply contains all the gates from $C$. The functions ${s'}$ and ${t'}$ are well-defined because gates are always the source and target of ${E_1}$- and ${E_2}$-edges, respectively. For convention, we use positive integers to tag elements of disjoint sets. With these considerations, we construct a control-driven circuit $\lambda$ by: ${V\cong E\sqcup E}$, ${U\cong G}$, ${O\cong E_1\sqcup E_1}$ and ${I\cong E_2\sqcup E_2}$. If ${\phi\colon V\rightarrow E\sqcup E}$, ${\psi\colon U\rightarrow G}$, ${\theta\colon O\rightarrow E_1\sqcup E_1}$ and ${\gamma\colon I\rightarrow E_2\sqcup E_2}$ are the corresponding isomorphisms, and ${\iota_j\colon E_j\sqcup E_j\rightarrow E\sqcup E}$ is the obvious inclusion function for ${j\in\{1,2\}}$, then: 
\begin{itemize}[noitemsep, topsep=4pt]
\item ${s\colon I\rightarrow V}$ is given by ${s(i)=(\phi^{-1}\circ\iota_2\circ\gamma)(i)}$,
\item ${\tau\colon I\rightarrow U}$ is given by ${\tau(i)=(\psi^{-1}\circ t'\circ\gamma)(i)}$,
\item ${t\colon O\rightarrow V}$ is given by ${t(o)=(\phi^{-1}\circ\iota_1\circ\theta)(o)}$,
\item ${\sigma\colon O\rightarrow U}$ is given by ${\sigma(o)=(\psi^{-1}\circ s'\circ\theta)(o)}$, and
\item ${c\colon V\rightarrow\Sigma}$ is given by $c(v) = 
    \begin{cases}
       \mathbbm{1} & v=\phi^{-1}(e,1)\text{ for some } (e,1)\in E\sqcup E\\
       \mathbbm{\mathbb{B}} & v=\phi^{-1}(e,2)\text{ for some } (e,2)\in E\sqcup E\\
    \end{cases}$
\end{itemize}
Strictly, the inverses of the isomorphisms produce a singleton set because they are injective. For the sake of argument, we consider inverses to yield single elements. 

Now, to show $\lambda$ fully adheres to Definition \ref{def:computon}, we need to prove that $c$, $\sigma$ and $\tau$ are surjective. For this, observe there is an inherent partition of $V$ induced by the equivalence relation ${v\sim w\iff \phi(v)=(e_1,j)\land \phi(w)=(e_2,j)}$ for all ${e_1,e_2\in E}$ and ${j\in\{1,2\}}$. So, by the def. of $c$, variables in the part ${j=1}$ are mapped to the unit type $\mathbbm{1}$ and variables in the part ${j=2}$ are sent to the Boolean type $\mathbb{B}$. As the partition evidently consists of two disjoint sets, it follows that $c$ is surjective.

To prove $\sigma$ is also surjective, we first show ${s'}$ also is by assuming for contradiction there is some ${g\in G}$ for which there is no ${(e,k)\in E_1\sqcup E_1}$ with ${s'(e,k)=g}$. As $g$ must necessarily have an outgoing edge because ${g^-=1}$, ${(g, m) \in E\setminus E_1}$ for some ${m\in N}$, i.e., ${g^+=0}$ as per the def. of ${E_1}$. Clearly, a contradiction to the property ${g^+=2}$ imposed by the definition of $G$. With this in mind, we now show ${\sigma=\psi^{-1}\circ s'\circ\theta}$ is onto by choosing an arbitrary unit ${u \in U}$. As ${\psi^{-1}}$ is surjective because it is an isomorphism, there must be some gate ${h \in G}$ with ${\psi^{-1}(h)=u}$. With the surjectivity of ${s'}$, we deduce the existence of an edge ${((h, n), l)\in E_1\sqcup E_1}$ where ${s'((h, n), l)=h}$ and ${l \in \{1,2\}}$. Given that $\theta$ is a isomorphism, ${\theta(o)=((h, n), l)}$ for some ${o \in O}$; thereby, proving $\sigma$ is onto. 

To prove ${\sigma\restriction_{(c\circ t)^{-1}(\mathbbm{1})}}$ is also onto, we use the definition of ${s'}$ to derive ${s'((h, n), 1)=h=s'((h, n), 2)}$ and simply choose ${((h, n), 1)\in E_1\sqcup E_1}$. Observing that ${c\circ t=c\circ\phi^{-1}\circ\iota_1\circ\theta}$, that ${\iota_1((h, n), 1)=((h, n), 1)}$ and that ${\phi^{-1}}$ is total (because $\phi$ is an isomorphism), we obtain ${\phi^{-1}(\iota_1(\theta(o)))=\phi^{-1}((h, n), 1)}$. By the definition of $c$, we derive ${c(\phi^{-1}((h, n), 1))=\mathbbm{1}}$, i.e., ${o \in (c\circ t)^{-1}(\mathbbm{1})}$. As ${\sigma(o)=(\psi^{-1}\circ s'\circ\theta)(o)=u}$ and ${o \in (c\circ t)^{-1}(\mathbbm{1})}$, we conclude ${\sigma\restriction_{(c\circ t)^{-1}(\mathbbm{1})}}$ is surjective. The surjectivities of $\tau$ and ${\tau\restriction_{(c\circ s)^{-1}(\mathbbm{1})}}$ follow analogously.

We just need to verify that there is some ${v\in V\setminus t(O)}$ with ${c(v)=\mathbbm{1}}$ and some ${w\in V\setminus s(I)}$ with ${c(w)=\mathbbm{1}}$. We show the former since the other is symmetric. For this, note that gates always exist as per Remark \ref{rem:circuits-always-gates} and that they always have fan-in $2$ as per Definition \ref{def:classical-circuit}. So, there is some edge ${(x,y)\in E}$ from an input variable ${x\in N}$ to a gate ${y\in N}$. Consequently, ${E\sqcup E}$ must contain both ${((x,y),1)}$ and ${((x,y),2)}$ so ${\phi^{-1}((x,y),1)\in V}$ because $\phi$ is an isomorphism. Assuming for contradiction that there is some ${o\in O}$ with ${t(o)=\phi^{-1}(\iota_1(\theta(o)))=\phi^{-1}((x,y),1)}$, we deduce ${((x, y), 1)\in E_1\sqcup E_1}$ given that ${\iota_1}$ is an inclusion, i.e., ${(x,y)\in E_1}$. By the definition of ${E_1}$, we have ${x^+\neq 0}$ which clearly contradicts that $x$ is an input variable with ${x^+=0}$. Hence, ${\phi^{-1}((x,y),1)\notin t(O)}$. So, by the definition of $c$, we conclude ${c(\phi^{-1}((x,y),1))=\mathbbm{1}}$, as required.\qed
\end{proof}

\begin{proposition}\label{prop:sound-transformation}
The control-driven circuit $\lambda$ of a NAND circuit $C$ is sound. 
\end{proposition}
\begin{proof}
If ${v\in V^+\cup s(I)}$, then ${\phi(v):=((m,n),j)\in E\sqcup E}$ exists for ${j\in \{1,2\}}$ and ${m,n\in N}$ so that ${(m,n)\in E}$. As it is easy to check that there is a path from every $C$-node to some output variable in $C$, the path ${\langle (m,n),\ldots,(x,y) \rangle}$ must exist for ${x,y\in N}$ and ${y^-=0}$. Supposing for contradiction ${\phi^{-1}((x,y),k)\notin V^-}$ for ${k\in \{1,2\}}$, there must be some ${i\in I}$ with ${s(i)=\phi^{-1}(\iota_2(\gamma(i)))=\phi^{-1}((x,y),k)}$. As ${\phi^{-1}}$ is injective because it is an isomorphism, ${\iota_2(\gamma(i))=((x,y),k)}$ so that ${(x,y)\in E_2}$ because $\iota_2$ is an inclusion. By the definition of $E_2$, this means (i) ${x^+\neq 2}$ or (ii) ${x^-\neq 1}$ or (iii) ${y^-\neq 0}$. If (i) or (ii) holds, then $x$ cannot be a gate so it has to be an input or an output variable. But this contradicts Definition \ref{def:classical-circuit} which says that there are no edges between outputs or from inputs to outputs. As (iii) immediately contradicts the fact that $y$ is an output, we conclude ${\phi^{-1}((x,y),k)\in V^-}$, as required.\qed
\end{proof}

\begin{remark}
The construction presented in the proof of Theorem \ref{th:classical-to-control} yields a control-driven Boolean circuit $\lambda$ in which there is a control variable for each edge of the corresponding NAND circuit $C$. This structural feature does not influence behaviour correspondence between $\lambda$ and $C$, since control values are always consumed and produced by computation units. 
\end{remark}

\begin{remark}\label{rem:equivalence-ins-outs}
Note that there are no shared input/output variables in $\lambda$ (similar to how linear logic treats resources). Particularly, if ${P:=\{n\in N\mid n^+=0\}}$ is the set of input variables of $C$, $\lambda$ has ${\sum_{n\in P}n^-}$ control invars and ${\sum_{n\in P}n^-}$ Boolean invars. Symmetrically, if ${Q:=\{n\in N\mid n^-=0\}}$ is the set of output variables of $C$, $\lambda$ has ${\sum_{n\in Q}n^+}$ control outvars and ${\sum_{n\in Q}n^+}$ Boolean outvars.
\end{remark}

Theorem \ref{th:classical-to-control} indicates that control-driven Boolean circuits are at least as powerful as their classical counterpart in terms of what can be computed. But they are more expressive because they are equipped with well-defined operators for the explicit definition of branching or iterative computations. These control flow constructs are not supported by classical Boolean circuits. The classical counterpart only supports (partial/total) sequencing and parallelising. 

Although they are more expressive, control-driven circuits can only compute Boolean functions on a fixed number of inputs. To address this, such a model can extend to families that non-uniformly perform computation, just as their classical counterpart. Such an extension is formalised in the proof of Theorem \ref{th:computons-any-boolean}. 

\begin{theorem}\label{th:computons-any-boolean}
For any Boolean function $\{0,1\}^*\rightarrow\{0,1\}$ there is a family of control-driven Boolean circuits able to compute it.
\end{theorem}
\begin{proof}
We know a family $\{C_k\}_{k\in\mathbb{N}}$ of NAND circuits can be defined to compute a Boolean function ${f\colon\{0,1\}^*\rightarrow\{0,1\}}$ where, for all ${k\in\mathbb{N}}$ and all ${x\in\{0,1\}^k}$, $C_k$ computes ${f(x)}$. By Theorem \ref{th:classical-to-control}, we know a circuit $C_k$ has a corresponding control-driven Boolean circuit $\lambda$. Therefore, it is possible to construct a family ${\{\lambda_k\}_{k\in\mathbb{N}}}$ of control-driven circuits where each $\lambda_k$ computes ${f(x)}$ for all ${x\in\{0,1\}^k}$.\qed 
\end{proof}

\begin{remark}
The only aspect to consider when computing a control-driven circuit $\lambda_k$ from a family $\{\lambda_k\}_{k\in\mathbb{N}}$ is that multiple copies of the same input value may be required because $\lambda_k$ does not admit shared invars. Copies must be precisely specified via the initial state $\delta^0$ of $\lambda_k$.
\end{remark}

\section{Related Work}
\label{sec:related-work}

In \cite{arellanes_compositional_2026}, we laid the theoretical foundations of the computon model, which form the basis for the present work. In that paper, we neither formalise behaviour for concrete computation units nor discuss any association with particular data types nor present definite operational semantics. In this paper, we move from the abstract to the particular by applying the computon model in the context of Boolean circuits. We realise this shift by treating units as families of NAND functions and by providing concrete execution semantics (see Definition \ref{def:cunit} and Section \ref{sec:dynamics}). We also provide a new semantics for sequencing so that operand soundness in the sense of Definition \ref{def:sound} is no longer required (see Definition \ref{def:sequential-computon}). Although the identity law is now preserved through this modification (cf. Theorem \ref{th:sequencing-identity}), Theorem \ref{th:sequencing-associativity} has to be proven again via Proposition \ref{prop:sequencing-mono} to guarantee that associativity still holds. For branching, unlike \cite{arellanes_compositional_2026}, we offer explicit proofs for Theorems \ref{th:branching-commutativity} and \ref{th:branching-associativity} via Proposition \ref{prop:adjoints-for-branching}. Although Definition \ref{def:branching} is equivalent to the original notion of branching, it now states the specific colimit operation to form such composites, derived from the existential proof presented in \cite{arellanes_compositional_2026}. Specific colimit operations for iterative composites are also given in Definitions \ref{def:head} and \ref{def:tail}. More importantly, in Section \ref{sec:equivalence}, we prove that our particular model instance is at least as powerful as families of Boolean circuits.
 
It is well-known that such a class of circuits only support sequencing and parallelising in their canonical form \cite{wegener_complexity_1987}.\footnote{Although Boolean circuits can be truly parallel, control flow is implicit so it is impossible to see which parts are sequential/parallel and how they are glued together.} To realise the simplest embodiment of branching, multiplexers (or demultiplexers) can be integrated for deterministically selecting inputs (or outputs) based on a so-called \emph{select line} \cite{lozhkin_complexity_2021,asadi_multiplexer_2024}. More advanced ad-hoc constructs can be introduced (e.g., noisy gates \cite{pippenger_networks_1985,mozeika_reliable_2015}, random flip-flops \cite{stipcevic_quantum_2016} or stochastic multiplexers \cite{brown_stochastic_2001,zhu_stochastic_2014}) to assign probability distributions over outputs with the aim of simulating probabilistic Turing machines (or random Boolean functions). Probabilistic logic has also proven suitable to deal with non-determinism in Boolean circuits \cite{mansinghka_stochastic_2008,chakrapani_probabilistic_2008,flaminio_conditional_2023}. 

Apart from the non-deterministic selection of computational paths, looping is another limitation of such a classical model. To address this, cyclic Boolean circuits \cite{riedel_cyclic_2012} have been proposed where gate outputs feed back as inputs to earlier gates. In circuits that incorporate memory elements (e.g., D-type flip-flops or latches), an output depends on current and previous inputs. In this case, the memory mechanism is realised through feedbacks that store and update internal states \cite{steinbach_sequential_2022,nishimura_classification_2023}. More recent work proposes to analyse cyclic circuits via categorical semantics by presenting them as morphisms in a freely generated symmetric traced monoidal category \cite{katis_feedback_2002,ghica_categorical_2016,ghica_diagrammatic_2017,di_lavore_canonical_2021,sprunger_differentiable_2025}. Syntactically, these categories are expressed as string diagrams \cite{piedeleu_introduction_2025} where a trace represents
a loop from morphism outputs to morphism inputs, leading to an intuitive graphical interpretation of an equational theory for circuits \cite{ghica_complete_2025}. 

Equational theories provide a formal framework for describing and reasoning about algebraic structures. In Boolean circuits, the fundamental composition operation comes from Boolean algebra, which allows the combination of logic gates with compatible bandwidths into more complex gates. As the output of one gate is the input of another \cite{mansinghka_stochastic_2008}, a composite corresponds to a total sequential gate. The law of abstraction \cite{turner_new_1979} defines modularity in the sense that total sequential gates can be further composed without the need of delving into their internals. 

To enable algebraic modularity, efforts have been made towards a full theory of compositional circuits. For example, \cite{lafont_towards_2003} proposes an algebra for circuit construction by treating gates as generators. Although it includes operations for total sequencing (by the usual composition of maps) and parallelising (over disjoint of maps), it does not support the algebraic construction of branching, partial sequential or iterative circuits. The works described in \cite{ghica_categorical_2016} and \cite{ghica_diagrammatic_2017} constitute a preliminary step towards filling this gap, but their treatment is ad-hoc and not fully formal since they rely on equations based on raw intuitions. To alleviate this, \cite{ghica_complete_2025} proposes to freely generate symmetric traced monoidal categories from circuit signatures. Arbitrary circuits can be composed using morphism composition as well as tensor and trace operations. Although total sequencing, parallelising and looping are addressed, there is no support for branching or partial sequencing. The reason for omitting partial sequencing is because the theory of symmetric monoidal categories does not explicitly separate data from control, as evidenced by the lack of separate wires (in string diagrams) for expressing control signals \cite{bonchi_string_2022}. In \cite{villoria_enriching_2025}, such categories are freely enriched over Eilenberg-Moore algebras to assign probabilities to wires to realise probabilistic branching. However, support for partial sequencing is still missing. Higher-dimensional automata \cite{fahrenberg_languages_2021,fahrenberg_kleene_2022} can accommodate partial sequencing, but only implicitly since they do not offer dedicated operators for this purpose. Control flow is implicit in the activation of state transitions and is not separated from data flow. Accordingly, independent and modular reasoning of control and data is not possible.

The modular, hierarchical structure of circuits may resemble works from category-theoretic cybernetics. For example, \cite{ehresmann_memory_2007} proposes the use of general colimits to bind components of a complex system in a hierarchical and bottom-up manner, albeit without any support for defining explicit control flow structures able to enact concrete causal and informational interactions. The work in \cite{arellanes_composition_2024} attempts to fill this gap by enabling the dynamic formation of spaces of total sequential composites, but without any colimit constructions. Sequential composition is implicit in the dynamical evolution of self-organising spaces.

Unlike the existing body of work, in our proposal there is no need of introducing ad-hoc circuit constructs (e.g., specialised gates) to realise sequencing, parallelising, branching or looping. Instead, control flow is achieved through algebraic operators that allow the compositional construction of complex circuits in a hierarchical, bottom-up manner. As a result, modularity is attained even in the presence of loops or non-determinism. Deterministic and non-deterministic circuits are compatible with each other. Due to separate constructs for control flow and data flow, independent reasoning of such dimensions is possible.

\section{Conclusions and Future Work}
\label{sec:conclusions}

In this paper, we introduced a (non-uniform) model for computing Boolean functions, referred to as (families of) control-driven Boolean circuits. We particularly propose colimit-based algebraic operators for the compositional construction of circuits into partial sequential, total sequential, parallel, branching or iterative composites. In Section \ref{sec:model}, we showed that the operators for parallelising and branching only satisfy commutativity and associativity, whereas partial sequencing only fulfils identity. Total sequencing obeys both identity and associativity. 

In Section \ref{sec:dynamics}, we explained that the execution of a control-driven circuit is governed by explicit control flow which defines the order in which computation units are evaluated. As a unit is a passive construct that synchronises control and Boolean values, it can only be evaluated when its input data is ready. By evaluation, we mean choosing a particular $k$-ary NAND function to be executed on the inputs; thereby, producing new Boolean values and control signals. Results governing the relation between composition and execution are part of future work. Particularly, we are interested in showing that termination can be guaranteed by construction by taking into account that composites are deadlock-free (cf. \cite{arellanes_compositional_2026}).

Section \ref{sec:equivalence} showed that every classical Boolean circuit can be converted into a control-driven one through the procedure given in the proof of Theorem \ref{th:classical-to-control}. Unfortunately, such a procedure yields a (top-level) composite $C$ whose (bottom-up) construction is not evident. Finding the minimal operations to form $C$ requires the definition of ``decomposition operators'' which we plan to provide in the near future. Another limitation of our procedure is that the number of control variables grows more than the number of edges of a classical circuit. In the future, we would like to alleviate this issue by offering an alternative approach to produce less complex circuits in terms of their size. For now, it suffices to just give a procedure to show that families of control-driven circuits are at least as powerful as their classical counterpart, i.e., they are capable of performing any deterministic algorithm (see Theorem \ref{th:computons-any-boolean}). In doing so, we indirectly prove that (compositionally constructed) control flow is inextricably present in ``any deterministic computation''. Although the proof builds upon the NAND basis $\{\uparrow\}$ to treat computation units homogeneously, other functionally complete sets can be used instead for extra expressivity, without changing our results. 

The direct implication of enabling explicit control is that, unlike classical Boolean circuits or any other classical model of computation (e.g., Turing machines or the lambda calculus), control-driven circuits explicitly determine the order in which (small or large) modular computations occur. Control-driven circuits are modular because there is no need of analysing their internal structure for composing them. To the best of our knowledge, this is the first comprehensive algebraic framework seamless integrating (total and partial) sequencing, parallelising, branching and iteration into Boolean circuits in a way that is well-defined and compositionally structured such as Boolean algebra itself. A characterisation of Boolean functions that can be represented via iteration is future work.

Although branching circuits are suitable to simulate randomness, they are not expressive enough to choose computational paths according to user-defined inputs. In the future, we plan to incorporate support for pre- and post-conditions into the semantics of computation units so as to realise user-defined branching. In this regard, we also plan to implement a visual programming language for the systematic construction of compositional control-driven circuits. For simulating probabilistic Turing machines, we expect to incorporate support for probabilistic choice; thereby, providing a more complete logical set to deal with the inherent randomness that (apparently) manifests in the physical world. 

\bibliographystyle{splncs04}
\bibliography{refs}
\newpage
\appendix

\section{Applications of Control-Driven Boolean Circuits}
\label{app:applications}

In this appendix, we present concrete examples of control-driven Boolean circuits, with the aim of demonstrating the applicability of our proposal. As a starting point, Figure \ref{fig:elementary-circuits} shows the equivalent circuits for the NOT, AND and OR functions. In Remark \ref{rem:nand-behaviour}, we mentioned that the $\uparrow_1$ operator is chosen from a family when there is exactly one Boolean input. As $\uparrow_1$ corresponds to NOT (see Definition \ref{def:cunit}), computation units reducing to that operation are annotated with the $\lnot$ symbol. To further enhance readability, units with ${k \geq 2}$ Boolean inputs are labelled with $\uparrow$, effectively denoting a $k$-ary NAND function.

\begin{figure}[H]
 \captionsetup{position=above}
 \centering 
 \subcaptionbox{NOT circuit}
 {
 \begin{tikzpicture}
 \computonPrimitive{1}{0}{0.6}{1.2}{$\lnot$}

 \flow{{0.2,1}}{$(0.2,1)+(0.8,0)$}{dashed}{};
 \node[draw=black,fill=white,inner sep=0pt,minimum size=3pt] (1q0) at (0.2,1) {}; 
 \din{1i1}{0.2}{0.6}{};

 \flow{{1.6,1}}{$(1.6,1)+(0.8,0)$}{dashed}{};
 \node[fill=black,inner sep=0pt,minimum size=3pt] (2q1) at (2.4,1) {};
 \dout{1o1}{1.6}{0.6}{};
 \end{tikzpicture}
 }
 \hspace{3.4cm}
 \subcaptionbox{AND circuit\label{fig:composite-and}}
 {
 \begin{tikzpicture}
\computonComposite{0.4}{-0.2}{4.4}{1.6};
\qmatch{3q}{2.5}{1}{};\flow{{1.45,1}}{{2.5,1}}{dashed}{};\flow{{2.6,1}}{{3.6,1}}{dashed}{};
\dmatch{3d}{2.5}{0.6}{}{above};\flow{{1.5,0.6}}{{2.5,0.6}}{}{};\flow{{2.56,0.6}}{{3.6,0.6}}{}{};

\computonPrimitive{1}{0}{0.6}{1.2}{$\uparrow$}
\qin{1q1}{0.2}{1}{};
\din{1i1}{0.2}{0.6}{};
\din{1i2}{0.2}{0.2}{};

\computonPrimitive{3.6}{0}{0.6}{1.2}{$\lnot$}
\qout{2q1}{4.2}{1}{};
\dout{2o1}{4.2}{0.6}{};
 \end{tikzpicture}
 }
 \hspace*{-65pt}
 \subcaptionbox{OR circuit\label{fig:composite-or}}
 {
 \begin{tikzpicture}
\computonComposite{0.4}{-0.4}{4.2}{3.4};
\computonComposite{0.8}{-0.2}{1.1}{3};

\computonPrimitive{1}{0}{0.6}{1.2}{$\lnot$}
\qin{1q1}{0.2}{1}{};
\din{1i1}{0.2}{0.6}{};

\computonPrimitive{1}{1.4}{0.6}{1.2}{$\lnot$}
\qin{2q1}{0.2}{2.4}{};
\din{2i1}{0.2}{2}{};

\qmatch{q1}{2.5}{1.9}{};\flow{{1.6,1.9}}{{2.5,1.9}}{dashed}{};\flow{{2.6,1.9}}{{3.6,1.9}}{dashed}{};
\dmatch{d1}{2.5}{1.5}{}{above};\flow{{1.6,1.5}}{{2.5,1.5}}{}{};\flow{{2.56,1.5}}{{3.4,1.5}}{}{};
\qmatch{q2}{2.5}{1}{};\flow{{1.6,1}}{{2.5,1}}{dashed}{};\flow{{2.6,1}}{{3.6,1}}{dashed}{};
\dmatch{d2}{2.5}{0.6}{}{above};\flow{{1.6,0.6}}{{2.5,0.6}}{}{};\flow{{2.56,0.6}}{{3.4,0.6}}{}{};

\computonPrimitive{3.4}{0.4}{0.6}{1.6}{$\uparrow$}
\qout{3q1}{4}{1.3}{};
\dout{3o1}{4}{0.9}{};
 \end{tikzpicture}
 }
 {
 \begin{tikzpicture}
\begin{scope}
\draw[draw=black,fill=white] (0,0.4) rectangle ++(0.12,0.25);
\node[inner sep=0pt,minimum size=0pt,label=right:{\scriptsize Computation unit}] at (0.2,0.5) {};
\node[draw=black,fill=white,inner sep=0pt,minimum size=3pt,label=right:{\scriptsize Control invar}] at (3.5,0.5) {};
\node[circle,draw=black,fill=white,inner sep=0pt,minimum size=3pt,label={right:{\scriptsize Boolean invar}}] at (6.8,0.5) {};
\end{scope}

\begin{scope}
\draw[dashed] (0,0) to node[pos=0.5]{\arrowflow} (0.6,0);
\node[inner sep=0pt,minimum size=3pt,label=right:{\scriptsize Control flow}] at (0.6,0) {};
\node[fill=black,inner sep=0pt,minimum size=3pt,label={right:{\scriptsize Control outvar}}] at (3.5,0) {};
\node[circle,fill=black,inner sep=0pt,minimum size=3pt,label={right:{\scriptsize Boolean outvar}}] at (6.8,0) {};
\end{scope}

\begin{scope}[yshift=-0.5cm]
\draw (0,0) to node[pos=0.5]{\arrowflow} (0.5,0);
\node[inner sep=0pt,minimum size=3pt,label=right:{\scriptsize Boolean flow}] at (0.5,0){};
\node[draw,fill=white,fill fraction={black}{0.5},inner sep=0pt,minimum size=3pt,label=right:{\scriptsize Control inoutvar}] at (3.5,0) {};
\node[circle,draw,fill=white,fill fraction={black}{0.5},inner sep=0pt,minimum size=3pt,label=right:{\scriptsize Boolean inoutvar}] at (6.8,0) {};
\end{scope}
\end{tikzpicture}
 }
 \caption{Elementary logical connectives as control-driven Boolean circuits.} 
 \label{fig:elementary-circuits}
\end{figure}
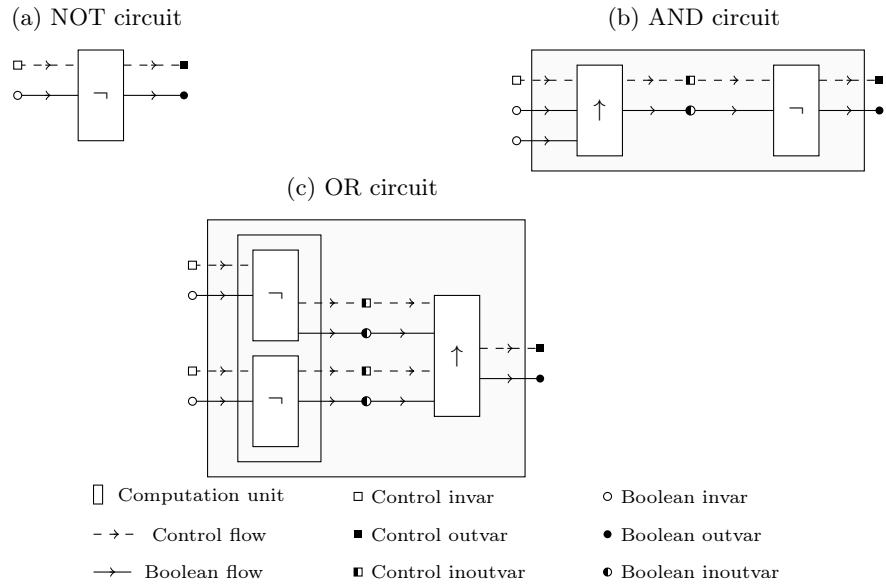

\vspace{-25pt}

In Figure \ref{fig:elementary-circuits}, it is clear that AND and OR take the form of composites, whereas NOT is just a primitive circuit consisting of a single computation unit with exactly one Boolean invar and one Boolean outvar. More precisely, AND is formed by composing a $2$-ary NAND circuit and a NOT primitive into a total sequential circuit. The OR composite is formed in two stages: First, two NOT primitives are put in parallel and then a total sequential composite is formed out of it together with a $2$-ary NAND function.

The elementary functions presented in Figure \ref{fig:elementary-circuits} can serve as building blocks to form even more complex circuits. For example, Figure \ref{fig:branching-example} displays a branching composite for a context-sensitive stochastic Boolean network \cite{liang_stochastic_2012}, which captures the oscillatory dynamics of a p53-Mdm2 regulatory pathway that is subject to molecular and genetic noise. The purpose of this circuit is to determine the next state of a gene by randomly selecting one out of four alternative sub-circuits. 

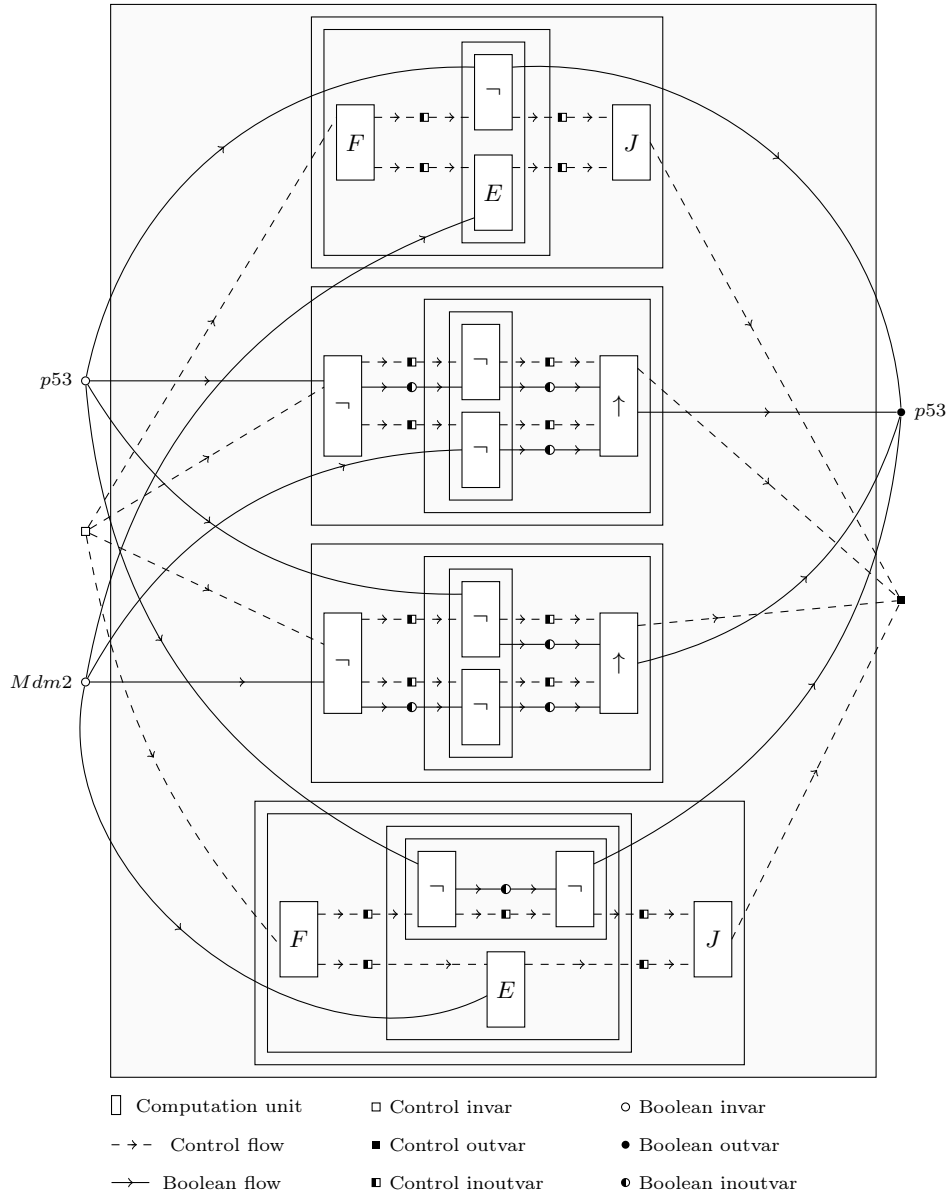
\begin{figure}[H]
 \centering 
 \begin{tikzpicture}[scale=0.83]

 \computonComposite{-2.6}{-0.6}{12.2}{17.1};

 
 \computonComposite{0.6}{12.3}{5.6}{4};
 \computonComposite{0.8}{12.5}{3.6}{3.6};
 \computonComposite{3}{12.7}{1}{3.2};
 
 \computonPrimitive{1}{13.7}{0.6}{1.2}{$F$} 
 \qmatch{f2o1}{2.4}{14.7}{};\flow{{1.6,14.7}}{f2o1}{dashed}{};\flow{f2o1}{{3.2,14.7}}{dashed}{};
 \qmatch{f2o2}{2.4}{13.9}{};\flow{{1.6,13.9}}{f2o2}{dashed}{};\flow{f2o2}{{3.2,13.9}}{dashed}{};
 \computonPrimitive{3.2}{14.5}{0.6}{1.2}{$\lnot$}
 \computonPrimitive{3.2}{12.9}{0.6}{1.2}{$E$} 
 \qmatch{n6o1}{4.6}{14.7}{};\flow{{3.8,14.7}}{n6o1}{dashed}{};\flow{n6o1}{{5.4,14.7}}{dashed}{};
 \qmatch{e2o1}{4.6}{13.9}{};\flow{{3.8,13.9}}{e2o1}{dashed}{};\flow{e2o1}{{5.4,13.9}}{dashed}{};
 \computonPrimitive{5.4}{13.7}{0.6}{1.2}{$J$} 
 
 
 \computonComposite{0.6}{8.2}{5.6}{3.8};
 \computonComposite{2.4}{8.4}{3.6}{3.4};
 \computonComposite{2.8}{8.6}{1}{3};

 \computonPrimitive{0.8}{9.3}{0.6}{1.6}{$\lnot$}
 
 \qmatch{g2o1}{2.2}{10.8}{};\flow{{1.4,10.8}}{g2o1}{dashed}{};\flow{g2o1}{{3,10.8}}{dashed}{};
 \dmatch{g2o2}{2.2}{10.4}{}{above};\flow{{1.4,10.4}}{g2o2}{}{};\flow{g2o2}{{3,10.4}}{}{};
 
 \qmatch{g2o3}{2.2}{9.8}{};\flow{{1.4,9.8}}{g2o3}{dashed}{};\flow{g2o3}{{3,9.8}}{dashed}{}; 

 \computonPrimitive{3}{8.8}{0.6}{1.2}{$\lnot$}
 \computonPrimitive{3}{10.2}{0.6}{1.2}{$\lnot$}

 \qmatch{n4o1}{4.4}{10.8}{};\flow{{3.6,10.8}}{n4o1}{dashed}{};\flow{n4o1}{{5.2,10.8}}{dashed}{};
 \dmatch{n4o2}{4.4}{10.4}{}{above};\flow{{3.6,10.4}}{n4o2}{}{};\flow{n4o2}{{5.2,10.4}}{}{};
 \qmatch{n5o1}{4.4}{9.8}{};\flow{{3.6,9.8}}{n5o1}{dashed}{};\flow{n5o1}{{5.2,9.8}}{dashed}{};
 \dmatch{n5o2}{4.4}{9.4}{}{above};\flow{{3.6,9.4}}{n5o2}{}{};\flow{n5o2}{{5.2,9.4}}{}{};

 \computonPrimitive{5.2}{9.3}{0.6}{1.6}{$\uparrow$}

 
 \computonComposite{0.6}{4.1}{5.6}{3.8};
 \computonComposite{2.4}{4.3}{3.6}{3.4};
 \computonComposite{2.8}{4.5}{1}{3};

 \computonPrimitive{0.8}{5.2}{0.6}{1.6}{$\lnot$}
 
 \qmatch{g1o1}{2.2}{6.7}{};\flow{{1.4,6.7}}{g1o1}{dashed}{};\flow{g1o1}{{3,6.7}}{dashed}{}; 
 
 \qmatch{g1o3}{2.2}{5.7}{};\flow{{1.4,5.7}}{g1o3}{dashed}{};\flow{g1o3}{{3,5.7}}{dashed}{};
 \dmatch{g1o4}{2.2}{5.3}{}{above};\flow{{1.4,5.3}}{g1o4}{}{};\flow{g1o4}{{3,5.3}}{}{};

 \computonPrimitive{3}{4.7}{0.6}{1.2}{$\lnot$}
 \computonPrimitive{3}{6.1}{0.6}{1.2}{$\lnot$}

 \qmatch{n2o1}{4.4}{6.7}{};\flow{{3.6,6.7}}{n2o1}{dashed}{};\flow{n2o1}{{5.2,6.7}}{dashed}{};
 \dmatch{n2o2}{4.4}{6.3}{}{above};\flow{{3.6,6.3}}{n2o2}{}{};\flow{n2o2}{{5.2,6.3}}{}{};
 \qmatch{n3o1}{4.4}{5.7}{};\flow{{3.6,5.7}}{n3o1}{dashed}{};\flow{n3o1}{{5.2,5.7}}{dashed}{};
 \dmatch{n3o2}{4.4}{5.3}{}{above};\flow{{3.6,5.3}}{n3o2}{}{};\flow{n3o2}{{5.2,5.3}}{}{};

 \computonPrimitive{5.2}{5.2}{0.6}{1.6}{$\uparrow$}
 
 
 \computonComposite{-0.3}{-0.4}{7.8}{4.2};
 \computonComposite{-0.1}{-0.2}{5.8}{3.8};
 \computonComposite{1.8}{0}{3.7}{3.4};
 \computonComposite{2.1}{1.6}{3.2}{1.6};
 
 \computonPrimitive{0.1}{1}{0.6}{1.2}{$F$} 
 \qmatch{f1o1}{1.5}{2}{};\flow{{0.7,2}}{f1o1}{dashed}{};\flow{f1o1}{{2.3,2}}{dashed}{};
 \qmatch{f1o2}{1.5}{1.2}{};\flow{{0.7,1.2}}{f1o2}{dashed}{};\flow{f1o2}{{4,1.2}}{dashed}{};
 \computonPrimitive{2.3}{1.8}{0.6}{1.2}{$\lnot$}
 \qmatch{n0o1}{3.7}{2}{};\flow{{2.9,2}}{n0o1}{dashed}{};\flow{n0o1}{{4.5,2}}{dashed}{};
 \dmatch{n0o2}{3.7}{2.4}{}{above};\flow{{2.9,2.4}}{n0o2}{}{};\flow{n0o2}{{4.5,2.4}}{}{};
 \computonPrimitive{4.5}{1.8}{0.6}{1.2}{$\lnot$} 
 
 \computonPrimitive{3.4}{0.2}{0.6}{1.2}{$E$} 
 \qmatch{nxo1}{5.9}{2}{};\flow{{5.1,2}}{nxo1}{dashed}{};\flow{nxo1}{{6.7,2}}{dashed}{};
 \qmatch{e1o1}{5.9}{1.2}{};\flow{{4,1.2}}{e1o1}{dashed}{};\flow{e1o1}{{6.7,1.2}}{dashed}{};
 \computonPrimitive{6.7}{1}{0.6}{1.2}{$J$}
 

 \dinplain{pin}{-3}{10.5}{$p53$};
 \flowdiag{pin}{{3.2,15.5}}{bend left=40}{}{pos=0.5,rotate=40};
 \flow{pin}{{{0.8,10.5}}}{}{};
 \flowdiag{pin}{{3,7.1}}{bend right=30}{}{pos=0.4,rotate=320};
 \flowdiag{pin}{{2.3,2.8}}{bend right=30}{}{pos=0.43,rotate=310}; 
 
 \qinplain{qin}{-3}{8.1}{}; 
 \flowdiag{qin}{{1,14.7}}{dashed}{}{pos=0.5,rotate=55};
 \flowdiag{qin}{{0.8,10.4}}{dashed}{}{pos=0.5,rotate=35};
 \flowdiag{qin}{{0.8,6.3}}{dashed}{}{pos=0.5,rotate=320};
 \flowdiag{qin}{{0.1,1.5}}{dashed,bend right=15}{}{pos=0.5,rotate=305};
 
 \dinplain{min}{-3}{5.7}{$Mdm2$};
 \flowdiag{min}{{3.2,13.1}}{bend left=30}{}{pos=0.92,rotate=35};
 \flowdiag{min}{{3,9.4}}{bend left=30}{}{pos=0.75,rotate=30}; 
 \flowdiag{min}{{0.8,5.7}}{}{}{pos=0.65};
 \flowdiag{min}{{3.4,0.7}}{bend right=65}{}{pos=0.45,rotate=320};

 \doutplain{pout}{9.2}{10}{$p53$};
 \flowdiag{{3.8,15.5}}{pout}{bend left=45}{}{rotate=320};
 \flow{5.8,10}{pout}{}{};
 \flowdiag{{5.8,6}}{pout}{bend right=30}{}{rotate=40};
 \flowdiag{{5.1,2.8}}{pout}{bend right=30}{}{rotate=50};

 \qoutplain{qout}{9.2}{7}{};
 \flowdiag{{6,14.3}}{qout}{dashed}{}{pos=0.4,rotate=305};
 \flowdiag{{5.8,10.7}}{qout}{dashed}{}{rotate=320};
 \flowdiag{{5.8,6.6}}{qout}{dashed}{}{pos=0.3,rotate=8};
 \flowdiag{{7.3,1.6}}{qout}{dashed}{}{rotate=55};
  
 \end{tikzpicture}
 \hspace{1cm}
 {
 \begin{tikzpicture}
\begin{scope}
\draw[draw=black,fill=white] (0,0.4) rectangle ++(0.12,0.25);
\node[inner sep=0pt,minimum size=0pt,label=right:{\scriptsize Computation unit}] at (0.2,0.5) {};
\node[draw=black,fill=white,inner sep=0pt,minimum size=3pt,label=right:{\scriptsize Control invar}] at (3.5,0.5) {};
\node[circle,draw=black,fill=white,inner sep=0pt,minimum size=3pt,label={right:{\scriptsize Boolean invar}}] at (6.8,0.5) {};
\end{scope}

\begin{scope}
\draw[dashed] (0,0) to node[pos=0.5]{\arrowflow} (0.6,0);
\node[inner sep=0pt,minimum size=3pt,label=right:{\scriptsize Control flow}] at (0.6,0) {};
\node[fill=black,inner sep=0pt,minimum size=3pt,label={right:{\scriptsize Control outvar}}] at (3.5,0) {};
\node[circle,fill=black,inner sep=0pt,minimum size=3pt,label={right:{\scriptsize Boolean outvar}}] at (6.8,0) {};
\end{scope}

\begin{scope}[yshift=-0.5cm]
\draw (0,0) to node[pos=0.5]{\arrowflow} (0.5,0);
\node[inner sep=0pt,minimum size=3pt,label=right:{\scriptsize Boolean flow}] at (0.5,0){};
\node[draw,fill=white,fill fraction={black}{0.5},inner sep=0pt,minimum size=3pt,label=right:{\scriptsize Control inoutvar}] at (3.5,0) {};
\node[circle,draw,fill=white,fill fraction={black}{0.5},inner sep=0pt,minimum size=3pt,label=right:{\scriptsize Boolean inoutvar}] at (6.8,0) {};
\end{scope}
\end{tikzpicture}
 }
 \caption{A branching circuit that captures the computational behaviour of the context-sensitive stochastic Boolean network presented in \cite{liang_stochastic_2012}. For increased readability, computation units are annotated to make evident the behaviour they perform. Here, the symbols $\lnot$, $\uparrow$, $E$, $F$ and $J$ denote NOT, NAND, Boolean eater, fork and join functions. All of them are well-defined as per Definitions \ref{def:cunit} and \ref{def:transition}.} 
 \label{fig:branching-example}
\end{figure}

For the first alternative, we first compose a parallel composite for the simultaneous activation of a NOT function and a ``Boolean eater''. In Section \ref{sec:dynamics}, we mentioned that a Boolean eater is a circuit that possesses a single computation unit with any number of Boolean invars and no Boolean outvars at all. Operationally, the unique unit chooses a function from the family but, since there are no Boolean outvars, only control signals are produced (see Definition \ref{def:transition}). This class of circuits is useful to wrap interfaces.

After composing our parallel composite, we form a partial sequential circuit by taking a ``fork'' as the left operand and the parallel structure as the right one. The result is then composed with a ``join'' primitive into a partial sequential composite that captures the whole behaviour of the first alternative branch.\footnote{The first alternative branch demonstrates the expressivity of our model to define computation units that do not receive or produce any Boolean values, but just fork or join control signals. Here, the unique computation units of the fork and join circuits are annotated with $F$ and $J$, respectively.}

The second and third alternatives are total sequential circuits for the invocation of a NOT primitive and an OR composite, in that order. The only difference is in the way they process Boolean information (cf. the leftmost NOT functions and the inner parallel composites). Finally, at the bottom-level, we have a composite which is similar to the first alternative. The only difference is that, rather than just computing NOT, the inner parallel composite activates a total sequential circuit of two NOT primitives, effectively passing information unchanged. In other words, the last branch defines a buffer to ``echo'' a $p53$ Boolean value.

It might not seem entirely obvious that the process of choosing one alternative or another is entirely governed by control flow. To elucidate this, observe that the initial state of our branching composite injects Boolean values into the $p53$ and $Mdm2$ variables as well as a control signal into the only control invar. As control flow arrows run from the control invar to the forks and to the leftmost NOT functions, all these units will be enabled initially (see Definition \ref{def:computation-unit-status}), but only one of them will be chosen for evaluation (as per Definition \ref{def:computation-unit-ready-reduction}). In other words, the initial state triggers only one alternative composite in the next execution step. 

Due to the modularity property of the proposed model, our branching composite can be used as a component to form even more complex circuits. To explain the modularity property in a more concrete manner, let us construct an iterative circuit to define the toggle action of a clocked set-reset flip-flop. First, let us reuse the buffer we use in the last alternative of our branching circuit, which we depict in Figure \ref{fig:buffer} for convenience. 

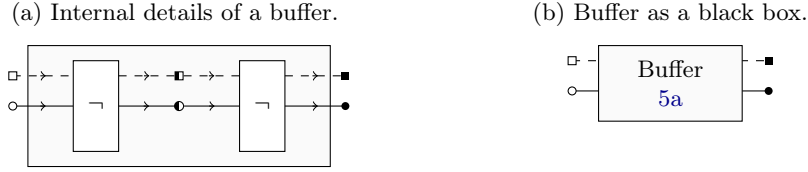
\begin{figure}[H]
 \captionsetup{position=above}
 \centering 
 \subcaptionbox{Internal details of a buffer.\label{fig:buffer}}
 {
 \begin{tikzpicture} 
 \computonComposite{1.7}{1.6}{4}{1.6}; 
 \qin{qin}{1.5}{2.8}{};
 \din{din}{1.5}{2.4}{};
 \computonPrimitive{2.3}{1.8}{0.6}{1.2}{$\lnot$}
 \qmatch{n0o1}{3.7}{2.8}{};\flow{{2.9,2.8}}{n0o1}{dashed}{};\flow{n0o1}{{4.5,2.8}}{dashed}{};
 \dmatch{n0o2}{3.7}{2.4}{}{above};\flow{{2.9,2.4}}{n0o2}{}{};\flow{n0o2}{{4.5,2.4}}{}{};
 \computonPrimitive{4.5}{1.8}{0.6}{1.2}{$\lnot$} 
 \qout{qout}{5.1}{2.8}{};
 \dout{dout}{5.1}{2.4}{};  
 \end{tikzpicture}
 }
 \hspace{1cm}
 \subcaptionbox{Buffer as a black box.}[0.45\textwidth]
 {
 \begin{tikzpicture} 
 \qin{qin}{1.5}{1.05}{};
 \din{din}{1.5}{0.65}{};
 \qout{qout}{3.3}{1.05}{};
 \dout{dout}{3.3}{0.65}{};  
 \computonComposite{1.85}{0.25}{1.9}{1};
 \node at (2.8,0.95){\footnotesize Buffer};
 \node at (2.8,0.55){\footnotesize\ref{fig:buffer}};
 \node at (2.8,0){};
 \end{tikzpicture}
 }
 \caption{A buffer is a total sequential circuit of two NOT primitives. Here, we demonstrate the modularity property of control-driven circuits. A similar black-boxing can be done with any of the composites displayed in Figures \ref{fig:elementary-circuits} and \ref{fig:branching-example}.} 
 \label{fig:modularity}
\end{figure}

This time we reuse this structure to build the partial sequential composite displayed in Figure \ref{fig:composite-entry}. As \ref{fig:composite-entry} will be the entry point to our iterative computation, its purpose is to trigger three buffers in parallel for echoing the values of three variables, namely $R$ (reset), $Q_n$ (current state) and $S$ (set).

These three values are intended to be consumed by the composite of Figure \ref{fig:composite-iteration}, which captures the behaviour of our main recurrent computation. Like \ref{fig:composite-entry}, this building block is constructed from existing structures; in this case, the AND and OR composites from Figure \ref{fig:elementary-circuits}. Operationally, \ref{fig:composite-iteration} realises the following equation to determine the next flip-flop state:

\begin{equation}\label{eq:next-state}
 Q_{n+1}=S~\lor~(\lnot R~\land~Q_n)   
\end{equation}

Fully realising Equation \ref{eq:next-state} requires the negation of $R$, before applying AND on the result and the current state $Q_n$. In other words, a partial sequential circuit needs to be formed to compute NOT and AND, in that order (see Figure \ref{fig:composite-iteration}). 

\begin{figure}[!h]
 \centering 
 \subcaptionbox{A composite to be used as the entry to our iterative computation. This is formed by first composing three buffers in parallel and then using a fork and the result as left and right operands for a partial sequential circuit. The resulting sequential circuit and a join primitive are then used as left and right operands to form an even more complex partial sequential structure.\label{fig:composite-entry}}
 {
 \begin{tikzpicture}[scale=0.75]

 \computonComposite{-0.5}{-0.6}{7.6}{4.7};
 \computonComposite{-0.1}{-0.4}{5}{4.3};
 \computonComposite{2.1}{-0.2}{2.3}{3.9};

 \qin{fin}{-0.7}{1.8}{};
 \dinplain{rin}{-0.7}{3.3}{$R$};\flow{rin}{{2.3,3.3}}{}{};
 \dinplain{qin}{-0.7}{2.5}{$Q_n$};\flowdiag{qin}{{2.3,2.05}}{bend left=20}{}{};
 \dinplain{sin}{-0.7}{0.2}{$S$};\flow{sin}{{2.3,0.2}}{}{};
 
 \computonPrimitive{0.1}{1.2}{0.6}{1.2}{$F$} 
 \qmatch{f1o1}{1.5}{2.7}{};\flowdiag{{0.7,1.8}}{f1o1}{dashed}{}{rotate=45};\flow{f1o1}{{2.3,2.7}}{dashed}{};
 \qmatch{f1o2}{1.5}{1.8}{};\flow{{0.7,1.8}}{f1o2}{dashed}{};\flow{f1o2}{{2.3,1.8}}{dashed}{};
 \qmatch{f1o3}{1.5}{0.8}{};\flowdiag{{0.7,1.8}}{f1o3}{dashed}{}{rotate=315};\flow{f1o3}{{2.3,0.8}}{dashed}{};
 
 \qmatch{nxo1}{5.1}{2.7}{};\flow{{4.3,2.7}}{nxo1}{dashed}{};\flowdiag{nxo1}{{5.9,1.8}}{dashed}{}{rotate=315};
 \qmatch{e1o1}{5.1}{1.8}{};\flow{{4.3,1.8}}{e1o1}{dashed}{};\flow{e1o1}{{5.9,1.8}}{dashed}{};
 \qmatch{e1o2}{5.1}{0.8}{};\flow{{4.3,0.8}}{e1o2}{dashed}{};\flowdiag{e1o2}{{5.9,1.8}}{dashed}{}{rotate=45};
 \computonPrimitive{5.9}{1.2}{0.6}{1.2}{$J$}

 \begin{scope}[xshift=2.3cm,yshift=2.5cm] 
    \computonComposite{0}{0}{1.9}{1};
    \node at (0.95,0.7){\footnotesize Buffer};
    \node at (0.95,0.3){\footnotesize\ref{fig:buffer}};
 \end{scope}
 \begin{scope}[xshift=2.3cm,yshift=1.25cm] 
    \computonComposite{0}{0}{1.9}{1};
    \node at (0.95,0.7){\footnotesize Buffer};
    \node at (0.95,0.3){\footnotesize\ref{fig:buffer}};
 \end{scope}
 \begin{scope}[xshift=2.3cm,yshift=0cm] 
    \computonComposite{0}{0}{1.9}{1};
    \node at (0.95,0.7){\footnotesize Buffer};
    \node at (0.95,0.3){\footnotesize\ref{fig:buffer}};
 \end{scope}

 \qout{fout}{6.5}{1.8}{};
 \doutplain{rout}{6.5}{3.3}{$R$};\flow{{4.2,3.3}}{rout}{}{};
 \doutplain{qout}{6.5}{2.5}{$Q_n$};\flowdiag{{4.2,2.05}}{qout}{bend left=20}{}{};
 \doutplain{sout}{6.5}{0.2}{$S$};\flow{{4.2,0.2}}{sout}{}{};

 \begin{scope}[xshift=10.5cm,yshift=1.2cm] 
    \qin{qin}{-0.35}{1.4}{};
    \din{rin}{-0.35}{1}{$R$};
    \din{qin}{-0.35}{0.6}{$Q_n$};        
    \din{sin}{-0.35}{0.2}{$S$};
    \qout{qout}{1.45}{1.4}{};
    \dout{rout}{1.45}{1}{$R$};
    \dout{qout}{1.45}{0.6}{$Q_n$};        
    \dout{sout}{1.45}{0.2}{$S$};
    \computonComposite{0}{0}{1.9}{1.6};
    \node at (0.9,1){\footnotesize Entry};
    \node at (0.9,0.6){\footnotesize\ref{fig:composite-entry}};
 \end{scope}
 \end{tikzpicture}
 }
 \hspace{1cm}
 \subcaptionbox{A composite to be used as the main iterative computation, which is formed by first composing NOT and AND into a partial sequential circuit and then using the result and OR as left and right operands, respectively, for a more complex partial sequential circuit.\label{fig:composite-iteration}}
 {
 \begin{tikzpicture}[scale=0.75]
 \computonComposite{0.2}{-0.4}{8.8}{2.8};
 \computonComposite{0.6}{0.2}{4.5}{2};
 
 \qin{cin}{0}{1.8}{};    
 \din{rin}{0}{1.4}{$R$};
 \dinplain{qin}{0}{0.6}{$Q_n$};\flow{qin}{{3,0.6}}{}{};
 \dinplain{sin}{0}{0}{$S$};\flowdiag{sin}{{6.5,0.6}}{bend right=16}{}{};
 \computonPrimitive{0.8}{0.8}{0.6}{1.2}{$\lnot$} 
 \qmatch{no1}{2.2}{1.8}{};\flow{{1.4,1.8}}{no1}{dashed}{};\flow{no1}{{3,1.8}}{dashed}{};
 \dmatch{no2}{2.2}{1.4}{}{below};\flow{{1.4,1.4}}{no2}{}{};\flow{no2}{{3,1.4}}{}{};
 \begin{scope}[xshift=3cm,yshift=0.4cm]     
    \computonComposite{0}{0}{1.9}{1.6};
    \node at (0.9,1){\small $\land$};
    \node at (0.9,0.6){\small\ref{fig:composite-and}};
 \end{scope}
 \qmatch{ao1}{5.7}{1.8}{};\flow{{4.9,1.8}}{ao1}{dashed}{};\flow{ao1}{{6.5,1.8}}{dashed}{};
 \dmatch{ao2}{5.7}{1.4}{}{below};\flow{{4.9,1.4}}{ao2}{}{};\flow{ao2}{{6.5,1.4}}{}{}; 
 \begin{scope}[xshift=6.5cm,yshift=0.4cm]     
    \computonComposite{0}{0}{1.9}{1.6};
    \node at (0.9,1){\small $\lor$};
    \node at (0.9,0.6){\small\ref{fig:composite-or}};
 \end{scope}
 \qout{cout}{8.4}{1.8}{};    
 \dout{qout}{8.4}{1.4}{$Q_{n+1}$};

 \begin{scope}[xshift=11.5cm,yshift=0.2cm] 
    \qin{qin}{-0.35}{1.4}{};    
    \din{rin}{-0.35}{1}{$R$};
    \din{qin}{-0.35}{0.6}{$Q_n$};        
    \din{sin}{-0.35}{0.2}{$S$};
    \qout{cout}{1.45}{1.4}{};
    \dout{qout}{1.45}{1}{$Q_{n+1}$};        
    \computonComposite{0}{0}{1.9}{1.6};
    \node at (0.95,1){\footnotesize Action};
    \node at (0.95,0.6){\footnotesize\ref{fig:composite-iteration}};
 \end{scope}
 \end{tikzpicture}
 }
 \hspace{1cm}
 \subcaptionbox{A composite to be used for the calculation of the next state, which simply is a partial sequential circuit for the execution of two NOT primitives.\label{fig:composite-next}}
 {
 \begin{tikzpicture}[scale=0.75]
 \node at (0,1.3){};
 \computonComposite{3.2}{1.3}{4}{1.9}; 
 \qin{qin}{3}{2.8}{};
 \din{din}{3}{2.4}{$Q_{n+1}$};
 \computonPrimitive{3.8}{1.8}{0.6}{1.2}{$\lnot$}
 \qmatch{n0o1}{5.2}{2.8}{};\flow{{4.4,2.8}}{n0o1}{dashed}{};\flow{n0o1}{{6,2.8}}{dashed}{};
 \dmatch{n0o2}{5.2}{2.4}{}{below};\flow{{4.4,2.4}}{n0o2}{}{};\flow{n0o2}{{6,2.4}}{}{};
 \computonPrimitive{6}{1.8}{0.6}{1.2}{$\lnot$} 
 \qout{cout}{6.6}{2.8}{};
 \dout{rout}{6.6}{2.4}{$R_{n+1}$};
 \dout{qout}{6.6}{2}{$Q_{n+1}$};
 \doutplain{sout}{6.6}{1.6}{$S_{n+1}$};\flowdiag{{4.4,2}}{sout}{bend right=12}{}{}; 

 \begin{scope}[xshift=12.2cm,yshift=1.5cm] 
    \qin{cin}{-0.35}{1.4}{};    
    \din{rin}{-0.35}{1}{$Q_{n+1}$};
    \qout{cout}{1.45}{1.4}{};
    \dout{rout}{1.45}{1}{$R_{n+1}$};
    \dout{qout}{1.45}{0.6}{$Q_{n+1}$};
    \dout{sout}{1.45}{0.2}{$S_{n+1}$};
    \computonComposite{0}{0}{1.9}{1.6};
    \node at (0.95,1){\footnotesize Next};
    \node at (0.95,0.6){\footnotesize\ref{fig:composite-next}};
 \end{scope}
 \end{tikzpicture}
 }
 \caption{Composite building blocks for the tail-iterative circuit depicted in Figure \ref{fig:iterative-example}. Internal structures are shown on the left and the corresponding modular black boxes are displayed on the right.} 
 \label{fig:iterative-blocks}
\end{figure}

As OR cannot be applied before finding ${\lnot R~\land~Q_n}$, the newly created composite needs to be invoked before passing its result to the OR circuit which, in addition, receives an $S$-value from the external world. In effect, this constitutes another partial sequential circuit, which we fully depict in Figure \ref{fig:composite-iteration}.

The result of \ref{fig:composite-iteration}, $Q_{n+1}$, is intended to be passed on to another circuit to determine the set and reset values for the next iteration. As we aim to implement a cyclic toggle action, we define $R_{n+1}=Q_{n+1}$ and $S_{n+1}=\lnot Q_{n+1}$, and simply echo $Q_{n+1}$ to the next iteration. Accordingly, we construct the partial sequential circuit of Figure \ref{fig:composite-next}, which behaves almost like a buffer, with the difference that $\lnot Q_{n+1}$ is sent to the external world immediately after computing it.

All the building blocks from Figure \ref{fig:iterative-blocks} and an eater primitive are used as operands for the colimit operation described in Definition \ref{def:tail}, effectively forming the tail-iterative circuit from Figure \ref{fig:iterative-example} which captures the behaviour of our cyclic toggle action. A glance at Definition \ref{def:tail} reveals that tail-iterative composites can only be created from sound circuits. By simply following the flows from invars to outvars in each building block, one can easily verify soundness. Eater primitives are always sound as per Proposition \ref{prop:functional-sound}. 

\begin{figure}[H]
 \centering 
 {
 \begin{tikzpicture}[scale=0.8]
 \computonComposite{0.2}{-0.2}{8.8}{4};
 \begin{scope}[xshift=0cm,yshift=1cm] 
    \qin{qin}{0}{1.4}{};    
    \din{rin}{0}{1}{$R$};
    \din{qin}{0}{0.6}{$Q_n$};    
    \din{sin}{0}{0.2}{$S$};
    \computonComposite{0.8}{0}{1.9}{1.6};
    \node at (1.7,1){\footnotesize Entry};
    \node at (1.7,0.6){\footnotesize\ref{fig:composite-entry}};
 \end{scope}
 \qmatch{eo3}{3.5}{2.4}{};\flow{{2.7,2.4}}{eo3}{dashed}{};\flowdiag{eo3}{{4.3,1.4}}{dashed}{}{rotate=310,pos=0.7}; 
 \dmatch{eo1}{3.5}{2}{}{below};\flow{{2.7,2}}{eo1}{}{};\flowdiag{eo1}{{4.3,1}}{}{}{rotate=310,pos=0.7};
 \dmatch{eo2}{3.5}{1.6}{}{below};\flow{{2.7,1.6}}{eo2}{}{};\flowdiag{eo2}{{4.3,0.6}}{}{}{rotate=310,pos=0.7};
 \dmatch{eo4}{3.5}{1.2}{}{below};\flow{{2.7,1.2}}{eo4}{}{};\flowdiag{eo4}{{4.3,0.2}}{}{}{rotate=310,pos=0.7};
 \begin{scope}[xshift=4.3cm,yshift=0cm]         
    \computonComposite{0}{0}{1.9}{1.6};
    \node at (0.95,1){\footnotesize Action};
    \node at (0.95,0.6){\footnotesize\ref{fig:composite-iteration}};
 \end{scope}
 \begin{scope}[xshift=4.3cm,yshift=2cm]
    \computonComposite{0}{0}{1.9}{1.6};
    \node at (0.95,1){\footnotesize Next};
    \node at (0.95,0.6){\footnotesize\ref{fig:composite-next}};
 \end{scope}
 \flowdiag{{4.3,3.4}}{eo3}{dashed}{}{rotate=225,pos=0.7};
 \flowdiag{{4.3,3}}{eo1}{}{}{rotate=225};
 \flowdiag{{4.3,2.6}}{eo2}{}{}{rotate=225,pos=0.25};
 \flowdiag{{4.3,2.2}}{eo4}{}{}{rotate=225,pos=0.25};

 \qmatch{ao1}{7}{2}{};
 \flowdiag{{6.2,1.2}}{ao1}{dashed}{}{rotate=45};
 \flowdiag{ao1}{{6.2,3}}{dashed}{}{rotate=135}; 
 \flow{ao1}{{7.8,2}}{dashed}{};
 \dmatch{ao2}{7}{1.6}{}{below};
 \flowdiag{{6.2,0.8}}{ao2}{}{}{rotate=45};
 \flowdiag{ao2}{{6.2,2.6}}{}{}{rotate=135};
 \flow{ao2}{{7.8,1.6}}{}{};
 \computonPrimitive{7.8}{1.2}{0.6}{1.2}{$E$}
 \qout{eo}{8.4}{1.8}{};
 \end{tikzpicture}
 }
 {
 \begin{tikzpicture}
\begin{scope}
\draw[draw=black,fill=white] (0,0.4) rectangle ++(0.12,0.25);
\node[inner sep=0pt,minimum size=0pt,label=right:{\scriptsize Computation unit}] at (0.2,0.5) {};
\node[draw=black,fill=white,inner sep=0pt,minimum size=3pt,label=right:{\scriptsize Control invar}] at (3.5,0.5) {};
\node[circle,draw=black,fill=white,inner sep=0pt,minimum size=3pt,label={right:{\scriptsize Boolean invar}}] at (6.8,0.5) {};
\end{scope}

\begin{scope}
\draw[dashed] (0,0) to node[pos=0.5]{\arrowflow} (0.6,0);
\node[inner sep=0pt,minimum size=3pt,label=right:{\scriptsize Control flow}] at (0.6,0) {};
\node[fill=black,inner sep=0pt,minimum size=3pt,label={right:{\scriptsize Control outvar}}] at (3.5,0) {};
\node[circle,fill=black,inner sep=0pt,minimum size=3pt,label={right:{\scriptsize Boolean outvar}}] at (6.8,0) {};
\end{scope}

\begin{scope}[yshift=-0.5cm]
\draw (0,0) to node[pos=0.5]{\arrowflow} (0.5,0);
\node[inner sep=0pt,minimum size=3pt,label=right:{\scriptsize Boolean flow}] at (0.5,0){};
\node[draw,fill=white,fill fraction={black}{0.5},inner sep=0pt,minimum size=3pt,label=right:{\scriptsize Control inoutvar}] at (3.5,0) {};
\node[circle,draw,fill=white,fill fraction={black}{0.5},inner sep=0pt,minimum size=3pt,label=right:{\scriptsize Boolean inoutvar}] at (6.8,0) {};
\end{scope}
\end{tikzpicture}
 }
 \caption{A tail-iterative circuit to compute the toggle action of a clocked SR flip-flop. Like any other iterative composite, this one is not built from sequencing, parallelising or branching. Instead, a self-contained colimit operation based on pushout and coproduct constructions is used at once, as described in Definition \ref{def:tail}. For this reason, we only show the building blocks and the top-level composite. Relative to the notation used in Definition \ref{def:tail}, the building blocks \emph{Entry}, \emph{Next}, \emph{E} and \emph{Action} correspond to the symbols $\lambda_2$, $\lambda_3$, $\lambda_4$ and $\lambda$, respectively.} 
 \label{fig:iterative-example}
\end{figure}
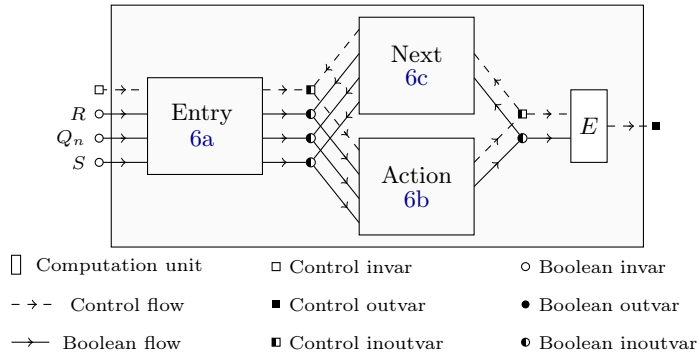

Table \ref{tab:iterative-table} shows the characteristic table of our iterative circuit, where it is clear that, having ${Q_n=S=R=0}$ as the initial state, will make our circuit toggle from the second iteration by triggering a set operation. If we rather start with $Q_n=1$ and $S=R=0$, the composite will also toggle from the second iteration. But this time, from a reset operation. Starting with an invalid state will always trigger reset in the second step. Once a set or a reset state is reached, the iterative circuit toggles until termination. Termination occurs randomly at any iteration, effectively simulating the moment clock signals are stopped being produced. 

\begin{table}[H]
\centering
\resizebox{\textwidth}{!}{
\begin{tblr}{
    colspec = {|c|c|c|c|c|c|c|}
  }
 \hline
 \footnotesize$S_n$ & \footnotesize$R_n$ & \footnotesize$Q_n$ & \footnotesize$S_{n+1}=\lnot Q_{n+1}$ & \footnotesize$R_{n+1}=Q_{n+1}$ & \footnotesize$Q_{n+1}=S\lor(\lnot R \land Q_n)$ & \\
 \hline
 \footnotesize0 & \footnotesize0 & \footnotesize0 & \footnotesize1 & \footnotesize0 & \footnotesize0 & \SetCell[r=2]{m}\footnotesize Hold\\
 \hline
 \footnotesize0 & \footnotesize0 & \footnotesize1 & \footnotesize0 & \footnotesize1 & \footnotesize1 & \\
 \hline
 \footnotesize0 & \footnotesize1 & \footnotesize0 & \footnotesize1 & \footnotesize0 & \footnotesize0 & \SetCell[r=2]{m} \footnotesize Reset\\
 \hline
 \footnotesize0 & \footnotesize1 & \footnotesize1 & \footnotesize1 & \footnotesize0 & \footnotesize0 & \\
 \hline
 \footnotesize1 & \footnotesize0 & \footnotesize0 & \footnotesize0 & \footnotesize1 & \footnotesize1 & \SetCell[r=2]{m}\footnotesize Set\\
 \hline
 \footnotesize1 & \footnotesize0 & \footnotesize1 & \footnotesize0 & \footnotesize1 & \footnotesize1 & \\
 \hline
 \footnotesize1 & \footnotesize1 & \footnotesize0 & \footnotesize0 & \footnotesize1 & \footnotesize1 & \SetCell[r=2]{m}\footnotesize Invalid\\
 \hline
 \footnotesize1 & \footnotesize1 & \footnotesize1 & \footnotesize0 & \footnotesize1 & \footnotesize1 & \\
 \hline
\end{tblr}
}
\caption{Characteristic table of the tail-iterative circuit from Figure \ref{fig:iterative-example}.}
\label{tab:iterative-table}
\end{table}

\vspace{-12pt}

Although our tail-iterative circuit checks for termination immediately after executing \ref{fig:composite-iteration}, it is possible to define a head-iterative composite instead. To accomplish this, we can use the same building blocks as operands. But this time, on the operator described in Definition \ref{def:head}.

\end{document}